\documentclass[12pt]{article}
\usepackage{etex}
\usepackage[utf8]{inputenc}
\usepackage[TS1,T1]{fontenc}
\usepackage{textcomp}
\usepackage[a4paper,tmargin=3cm,bmargin=3cm,rmargin=2.2cm,lmargin=2.2cm]{geometry}
\usepackage{lmodern}
\usepackage[normalem]{ulem}

\usepackage{amsfonts,amssymb,amsmath,amscd,mathtools,cite,slashed}
\usepackage[amsthm,thmmarks,amsmath]{ntheorem}

\usepackage{color}
\usepackage{psfrag}
\usepackage{epsfig}
\usepackage{pdflscape}
\usepackage{hyperref}
\usepackage{verbatim} 
\usepackage{graphicx}
\usepackage{enumitem}
\usepackage{xcolor}
\usepackage{framed}
\definecolor{shadecolor}{gray}{0.8}
\definecolor{lgray}{gray}{0.5}

\DeclareMathOperator{\Res}{Res}         
\DeclareMathOperator{\Tr}{Tr}             
\DeclareMathOperator{\sign}{sign}         
\DeclareMathOperator{\oo}{\scriptstyle{\mathcal{O}}}             
\DeclareMathOperator{\OO}{\mathcal{O}}             
\newcommand{\oz}{\oo_0}             
\newcommand{\Oz}{\OO_0}             
\newcommand{\oinf}{\oo_{\infty}}             
\newcommand{\Oinf}{\OO_{\infty}}             
\DeclareMathOperator{\Dom}{Dom}			

\newcommand{\C}{\mathbb{C}}              
\newcommand{\DD}{\mathcal{D}}           
\renewcommand{\H}{\mathcal{H}}          

\newcommand{\N}{\mathbb{N}}            
\newcommand{\floor}[1]{\lfloor#1\rfloor} 
\newcommand{\R}{\mathbb{R}}             
\newcommand{\wt}{\widetilde}                 
\newcommand{\Z}{\mathbb{Z}}                 
\renewcommand{\.}{\cdot}                         
\newcommand{\vc}{\vcentcolon =}             

\def\<#1,#2>{\langle#1\,,\,#2\rangle}      
\newbox\ncintdbox \newbox\ncinttbox
\setbox0=\hbox{$-$} \setbox2=\hbox{$\displaystyle\int$}
\setbox\ncintdbox=\hbox{\rlap{\hbox
    to \wd2{\kern-.1em\box2\relax\hfil}}\box0\kern.1em}
\setbox0=\hbox{$\vcenter{\hrule width 4pt}$}
\setbox2=\hbox{$\textstyle\int$}
\setbox\ncinttbox=\hbox{\rlap{\hbox
    to \wd2{\kern-.14em\box2\relax\hfil}}\box0\kern.1em}

\newcommand{\oh}{{\tfrac{1}{2}}}          

\renewcommand{\<}{\left\langle}		
\renewcommand{\Re}{\mathfrak{R}}			
\newcommand{\zP}{\zeta_{P}}
\newcommand{\PP}{\mathfrak{P}}

\newcommand{\Rez}[1]{\underset{#1}{\Res} \;}		

\def\R{\mathbb{R}}
\def\C{\mathbb{C}}
\def\N{\mathbb{N}}
\def\Z{\mathbb{Z}}

\newcommand{\Mellin}{\mathcal{M}}
\newcommand{\Fourier}{\mathcal{F}}

\DeclareMathOperator{\heat}{htr}
\newcommand{\abs}[1]{\left\lvert#1\right\rvert}

\def\oh{\frac{1}{2}}

\newcommand{\tzero}{\; \underset{t \downarrow 0}{\sim} \;}	

\newcommand{\fa}{\mathrm{deg} \, A}
\newcommand{\fb}{\mathrm{deg} \, B}
\newcommand{\Rnn}{\R^+ \cup \{0\}}		
\newcommand{\Np}{\N^+}
\newcommand{\ZZ}{\mathcal{Z}}								
\newcommand{\trunc}[1]{\hat{#1}}
\newcommand{\zT}{\trunc{\zeta}_{P}^{N}}

\newcommand{\cv}{= \vcentcolon}  
\DeclareMathOperator{\Li}{Li}		
\def \bangle{ \atopwithdelims \langle \rangle}

\newtheorem{assumption}{Assumption}[section]
\newtheorem{theorem}[assumption]{Theorem}

\newtheorem{Cor}[assumption]{Corollary}

\newtheorem{lem}[assumption]{Lemma}

\newtheorem{defn}[assumption]{Definition}
\newtheorem{Def}[assumption]{Definition}
\newtheorem{definition}[assumption]{Definition}
\newtheorem{prop}[assumption]{Proposition}
\newtheorem{Prop}[assumption]{Proposition}

\newtheorem{remark}[assumption]{Remark}
\newtheorem{example}[assumption]{Example}
\setlist[enumerate,1]{label=\upshape(\textit{\roman*})}

\begin{document}

{
\makeatletter\def\@fnsymbol{\@arabic}\makeatother 

\title{Asymptotic and exact expansions of heat traces}
\author{
Micha\l\  Eckstein\footnote{Faculty of Physics, Astronomy and Applied Computer Science, Jagellonian University}, 
Artur Zaj\k{a}c\footnote{Faculty of Mathematics and Computer Science, Jagellonian University}
}


\maketitle
}

\begin{abstract}
We study heat traces associated with positive unbounded operators with compact inverses. With the help of the inverse Mellin transform we derive necessary conditions for the existence of a short time asymptotic expansion. The conditions are formulated in terms of the meromorphic extension of the associated spectral zeta-functions and proven to be verified for a large class of operators. We also address the problem of convergence of the obtained asymptotic expansions. General results are illustrated with a number of explicit examples.
\end{abstract}


\section{Introduction}

Given a positive, possibly unbounded, operator $P$ with a compact resolvent, acting on a separable infinite-dimensional Hilbert space $\H$ one can define the associated heat operator $e^{-t\, P}$ for $t >0$. The latter, under some mild conditions on $P$, is trace class for any $t>0$. It turns out that a close inspection of its kernel, i.e. the function $t \mapsto \Tr e^{-t\, P}$, reveals a lot of information of geometrical nature. For instance, if $P$ is a differential operator of Laplace type defined on a closed Riemannian manifold then the classical results \cite{Gilkey1} show that there exists an asymptotic expansion of the form
\begin{align} \label{GilkeyHeat}
\Tr e^{-t\, P} \tzero \sum_{k \geq 0} a_k(P) t^{(k-d)/2},
\end{align}
where $d$ is the dimension of the manifold and $a_k$'s -- called Seeley-deWitt coefficients -- are given by the integrals over the manifold of some geometrical invariants. Moreover, the coefficients $a_k$ can be expressed as
\begin{align*}
a_k(P) = \Rez{s=(d-k)/2} \Gamma(s) \zP(s),
\end{align*}
where $\zP(s) \vc \Tr P^{-s}$ is the (meromorphic extension of the) spectral zeta-function associated with $P$.

The existence of an asymptotic expansion of $\Tr e^{-t\, P}$ was proven for $P$ being a classical positive elliptic pseudodifferential operator of order $m \in \N$ (see \cite{GilkeyGrubb} and references therein). In this case
\begin{align*}
\Tr  e^{-t P} \tzero \sum_{k = 0}^\infty a_k(P) t^{(-d+k)/m} + \sum_{l=1}^{\infty} b_l(P) t^l \log t
\end{align*}
and
\begin{align*}
a_k(P) & = \Rez{s=(d-k)/m} \Gamma \left( s \right) \zeta_{P}(s),\\
b_l(P) & = - \Rez{s=-l} (s+l) \Gamma \left( s \right) \zeta_{P}(s).
\end{align*}
In fact, this result can be extended to certain classes of \textit{nonclassical} pseudodifferential operators. In \cite{Lesch} for instance, the heat traces of pseudodifferential operators with log-polyhomogeneous symbols have been studied.

With the development of noncommutative geometry \cite{ConnesNCG}, the need came to investigate heat traces associated with positive functions of Dirac operators in the framework of spectral triples. Unfortunately, there is no analogue of the formula \eqref{GilkeyHeat} for a general spectral triple \cite{Vassilevich2}. In fact, the existence of an asymptotic expansion of the heat trace is assumed whenever needed in applications (see for instance \cite[Section 11]{ConnesMarcolli}, \cite[Section 2.1]{ConnesModular}) and has been proven rigorously only for a few specific examples \cite{HeatNCT,Vassilevich3,PodlesSA}. Whence the results of \cite{HeatNCT,Vassilevich3} essentially used the methods of pseudodifferential calculus, the casus of the standard Podle\'s sphere \cite{PodlesSA} required completely different tools (see Section \ref{sec:exp}).

We note that the interplay between heat traces and spectral zeta-functions has been investigated in a very general framework of von Neumann algebras by a number of authors \cite{ConnesNCG,DixZeta6,DixZeta4,DixZeta5}. However, the efforts of the latter focused on the leading behaviour of heat traces and its consequences for measurability.

The purpose of this paper is to study heat traces associated with general densely defined positive operators $P$ with compact inverses. In particular, we formulate sufficient conditions for the existence of a small $t$ asymptotic expansion of $\Tr e^{-t\, P}$. Having an asymptotic expansion at hand, a natural question one may pose is that of its convergence. This issue has not been studied in full generality even in the context of $P$ being a differential operator (see however \cite{PodlesSA,ILVGlobal}). We show how the conditions on $P$ shall be refined in order to get an exact formula for the heat trace valid on some open interval $(0,T)$. The motivation behind our work comes from noncommutative geometry, but the framework of the studies is even wider.

The heat trace methods have multifarious applications in theoretical physics (see \cite{VassilevichReport} for a review). They are in common use in quantum field theory \cite{HeatQFT,Anomalies,Elizalde89,Fulling}, also in its noncommutative version \cite{Wulkenhaar}. In general, one only disposes of an asymptotic expansion of the heat trace as  $t \downarrow 0$. This implies that the field-theoretic calculations performed with the help of this method are only perturbative. Needless to say that a control on the convergence of a perturbative expansion is of crucial importance.

In noncommutative geometry, the heat trace is the cornerstone of bosonic spectral action computations \cite{ConnesMarcolli,ConnesSA,ILVGlobal,ILVTorsion,ILVWeak}. The large energies expansion of the latter is based on the asymptotic expansion of the heat trace associated with the relevant Dirac operator. Recently, also the exact computations of the spectral action gained interest \cite{ConnesFLRW,PhD,PodlesSA,Marcolli1,Marcolli2,Piotrek1,TehPHD} due to their possible application to the study of cosmic topology.

The plan of the paper presents itself as follows: In Section \ref{Prelim} we recollect some basic notions on spectral functions associated with positive operators. Then, in Section \ref{sec:general} we discuss in details the interplay between the meromorphic extension of the spectral zeta function $\zP$ and the asymptotic expansion of the associated heat trace $\Tr e^{-t\, P}$, by gathering results on general Dirichlet series \cite{Hardy,Hardy_div} and the Mellin transform \cite{Flajolet,Paris}. Moreover, we present a set of sufficient assumptions on $P$ so that the associated heat trace is controlled for $t$ in some non-empty open interval. Section \ref{sec:examples} illustrates the general theorems with various special cases and examples coming from Dirac-type operators on both classical manifolds and noncommutative spaces. 
We end with an outlook on the possible generalisations and applications of our results. We also discuss the limitations of the method and compare its usefulness with the Tauberian theorems commonly used in this domain.

\section{Preliminaries}\label{Prelim}

\subsection{Notations}

Let us first fix some notations:
\begin{itemize}
\item $\N$ denotes the non-negative integers, $\Np$ the positive ones, $\Z^*$ stands for the non-zero integers and $\R^+$ for positive reals.
\item Unless stated otherwise, $t$ will always denote a positive parameter and $s$ a complex one.
\item $f(s) \approx g(s)$ means that $\lim_{\abs{s} \to \infty} \, f(s)/g(s) = 1$.
\item $f(x) = \OO_{x\to x_0} \left( g(x) \right)$ means that $\limsup_{x \to x_0} \abs{f(x)/g(x)} < \infty$, for $x,x_0 \in \R$. The notation $f(x) = \OO_{x_0} \left( g(x) \right)$ will be used when the variable is obvious.
\item $f(x) = \oo_{x\to x_0} \left( g(x) \right)$ means that $\limsup_{x \to x_0} \abs{f(x)/g(x)} = 0$, for $x,x_0 \in \R$. The notation $f(x) = \oo_{x_0} \left( g(x) \right)$ will be used when the variable is obvious.
\item $f(t) \tzero \sum_n \phi_n(t)$ denotes an asymptotic expansion (see Definition \ref{DefAsymp}) of $f$ as $t$ tends to 0 from above.
\item Unless stated otherwise, $P$ will be a positive densely defined operator with a compact inverse, acting on a separable infinite-dimensional Hilbert space.
\end{itemize}

\subsection{Heat traces}

\begin{definition}
The heat trace of the operator $P$ is the function $\heat_P:\R^+\to\R^+$, defined as
\begin{equation*}
\heat_P(t)\vc\Tr e^{-t\,P}.
\end{equation*}
We say that the heat trace is \emph{well-defined} if  $e^{-t \,P}$ is a trace class operator for any $t>0$.
\end{definition}

As $P$ has a compact inverse, its spectrum $\sigma(P)$ is a discrete subset of $\R$, which can be ordered into a sequence increasing to infinity 
\[
\sigma(P)=(\lambda_n)_{n=0}^\infty\,, \qquad 0<\lambda_0<\lambda_1<\ldots, \qquad \lim_{n\to\infty}\lambda_n=\infty.
\]
We will denote the multiplicity of the eigenvalue $\lambda_n\in\sigma(P)$ by $M_n$. For further purposes, we also define the spectral growth function as
\begin{align}\label{spectral_growth}
N(\lambda)\vc\sum_{\{n\,:\,\lambda_n\leq\lambda\}} M_n .
\end{align}

With these conventions, the heat trace can be written as
\begin{equation}
\label{hk_sum}
\heat_P(t)=\sum_{n=0}^{\infty}M_n\,e^{-t\,\lambda_n}.
\end{equation}

The sum is of the form of a general Dirichlet series, which is defined as
\begin{equation}
\label{gen_Dirichlet_series}
\sum_{n=0}^\infty a_n \,e^{-s\,b_n},
\end{equation}
for $s$ in some (possibly empty) subset of $\C$, $a_n\in \C$ and $(b_n)_{n=0}^{\infty}$ a sequence of real numbers increasing to infinity. The convergence of \eqref{gen_Dirichlet_series} is governed by the following theorem:

\begin{theorem}[\cite{Hardy}, Theorem 7 with the footnote ]
\label{thm:abscissa}
If $\sum_{n=0}^\infty a_n=\infty$ then the general Dirichlet series \eqref{gen_Dirichlet_series} converges for $\Re (s)>L$ and diverges for $\Re (s)<L$, where $L$ is given by
\begin{equation*}
L=\limsup_{n\to\infty}\, b_n^{-1}\,\log\left(a_0+\ldots+a_n\right)
\end{equation*}
and $L\geq0$.
\end{theorem}

We will call such $L$ the \emph{abscissa of convergence} of the general Dirichlet series \eqref{gen_Dirichlet_series}. The inequality $L\geq0$ follows from the fact that for $s=0$ \eqref{gen_Dirichlet_series} is equal to $\sum_{n=0}^\infty a_n$, which is divergent. Note that the abscissa can be infinite, what means that the series is nowhere convergent. 

\begin{Prop}
\label{cor:hk_abscissa}
The heat trace of the operator $P$ is well-defined:
\begin{enumerate}
\item \label{item:hk_abscissa2} if and only if $N(\lambda_n)=\OO_{n\to\infty}(e^{\epsilon\lambda_n})$ 
for any positive $\epsilon$;
\item \label{item:hk_abscissa3} if there exist $\alpha \geq 0$ such that $M_n =\Oinf(n^\alpha)$ and $\log n=\oinf(\lambda_n)$ (i.e. $\lambda_n$ grow faster than $\log n$).
\end{enumerate}
\end{Prop}

\begin{proof}
In order to have $\heat_P$ well-defined, we need the abscissa of convergence $L$ of the series \eqref{hk_sum} to be 0. Taking $a_n=M_n, b_n=\lambda_n$ for all $n\in\N$ we see that $\sum_{n=0}^\infty a_n=\infty$ as $P$ has an infinite number of eigenvalues. Now, by Theorem \ref{thm:abscissa} we get that $\heat_P$ is well-defined iff $L=\limsup_{n\to\infty}\,\lambda_n^{-1}\log N(\lambda_n)$ equals to 0. This is equivalent to the statement that for any $\epsilon>0$ there exists $n_0$ such that for any $n\geq n_0$ we have $\lambda_n^{-1}\log N(\lambda_n) \leq \epsilon$ or, equivalently, $N(\lambda_n) \leq e^{\epsilon\lambda_n}$. This in turn is equivalent to $N(\lambda_n)=\Oinf(e^{\epsilon\lambda_n})$ for any $\epsilon>0$ and thus \ref{item:hk_abscissa2} follows.

For statement \ref{item:hk_abscissa3} let us first note that since $M_n = \Oinf(n^{\alpha})$ for some $\alpha\geq0$, then $N(\lambda_n)=\Oinf(n^{\alpha+1})$. Take a positive constant $C$ such that $N(\lambda_n)\leq Cn^{\alpha+1}$ for $n\geq N$ for some $N\in\N$. Then
\begin{align*}
L&=\limsup_{n\to\infty}\,\lambda_n^{-1} \log\left(N(\lambda_n)\right)
\leq\limsup_{n\to\infty}\,\lambda_n^{-1} \log\left(Cn^{\alpha+1}\right)\\
&=\limsup_{n\to\infty}\,\lambda_n^{-1} \left(\log C+(\alpha+1)\log n \right),
\end{align*}
what tends to 0 under the hypothesis of \ref{item:hk_abscissa3}.
\end{proof}

Note that the condition in \ref{item:hk_abscissa3} is only sufficient: consider e.g. $\lambda_n=n^2,\text{ and } M_n=2^n$ which satisfies the hypothesis of \ref{item:hk_abscissa2} but not the one of \ref{item:hk_abscissa3}.

To conclude this subsection, we remark that heat traces can be defined for operators bounded from below, possibly with a non-trivial kernel. However, the zeta-functions to be described below apply only to positive invertible operators (see, however, Section \ref{sec:truncated}).

\subsection{Spectral zeta-functions}

\begin{definition}
The \emph{zeta-function} associated with the operator $P$ is a complex function
\begin{equation*}
\C \supset \Dom(\zeta_P) \ni s \mapsto \zeta_P(s)=\Tr P^{-s}.
\end{equation*}
We say that $\zeta_P$ is \emph{well-defined} if $\Dom(\zeta_P)$ is non-empty.
\end{definition}

Using the spectral theorem this can be written as
\begin{equation}
\label{zeta_sum}
\zeta_P(s)=\sum_{n=0}^{\infty} M_n\,\lambda_n^{-s},
\end{equation}
which is again a general Dirichlet series \eqref{gen_Dirichlet_series} with $a_n = M_n$ and $b_n = \log \lambda_n$. Using Theorem \ref{thm:abscissa} we get

\begin{Prop}
\label{cor:zeta_abscissa}
The abscissa of convergence of the zeta-series \eqref{zeta_sum} is given by 
\begin{equation}
\label{zeta_abscissa_inf}
L=\inf\{\alpha\in\R: N(\lambda_n)=\OO_{n\to\infty}(\lambda_n^\alpha)\}.
\end{equation}
\end{Prop}
Note that it may happen that $L=+\infty$, which means that the zeta-function is not well-defined (e.g. when $\lambda_n=\log n, M_n=1$ for $n\geq2$, then $N(\lambda)\approx\exp\lambda$).

\begin{proof}
Take any $\alpha$ such that there exists a constant $c$ satisfying $N(\lambda_n)\leq c\,\lambda_n^\alpha$. Then by setting $a_n=M_n, b_n=\log\lambda_n$ we have by Theorem \ref{thm:abscissa}
\begin{equation}
\label{limsup_lambda}
L=\limsup_{n\to\infty}\tfrac{\log N(\lambda_n)}{\log\lambda_n}\leq\limsup_{n\to\infty}\tfrac{\log (c\,\lambda_n^\alpha)}{\log\lambda_n}=\alpha.
\end{equation}
Thus $L\leq K$, where $K$ denotes the RHS of \eqref{zeta_abscissa_inf}.

On the other hand, if we suppose that $L<K$ then there exists $\alpha$ such that $L<\alpha<K$. Using first equality of \eqref{limsup_lambda} we can find such $n_0\in\N$ that $\log N(\lambda_n) (\log\lambda_n)^{-1}<\alpha$ for any $n\geq n_0$. Then we get that $N(\lambda_n)<\lambda_n^\alpha$ for $n\geq n_0$, which contradicts the assumption $\alpha<K$. Thus $L=K$ and the proposition is proved.
\end{proof}

\begin{Prop}
\label{cor:hk_order}
If an operator $P$ is such that its zeta-function is well-defined with finite abscissa of convergence $L$, then its heat trace is also well-defined and
\begin{equation*}
\heat_P(t)=\Oz(t^{-\alpha})\qquad\text{for all }\alpha>L.
\end{equation*}
\end{Prop}
\begin{proof}
Comparing Proposition \ref{cor:hk_abscissa} \ref{item:hk_abscissa2} with Proposition \ref{cor:zeta_abscissa} we get that $\heat_P$ is well-defined. 

Now, let us take any $\alpha > L \geq 0$. Then, there exists a positive constant $C(\alpha)$ such that
\begin{align*}
x^{\alpha} e^{-x} \leq C(\alpha),
\end{align*}
for any $x >0$, as the function $x \mapsto x^{\alpha} e^{-x}$ is bounded on $\Rnn$ for any $\alpha \geq 0$.
Therefore, for any $\alpha > L$ we have
\begin{align*}
0 \leq t^{\alpha} \heat_P(t) = \sum_{n=0}^{\infty} M_n \, t^{\alpha} e^{-t \lambda_n} \leq C(\alpha) \sum_{n=0}^{\infty} M_n \lambda_n^{-\alpha} = C(\alpha) \zeta_P(\alpha) < \infty.
\end{align*}
Hence $\heat_P(t)=\Oz(t^{-\alpha})$.
\end{proof}

\subsection{Mellin transform}

\begin{definition} (see \cite{Paris} for instance)
The Mellin transform of a locally Lebesgue integrable function $f$ defined over $\R^+$ is a complex function $\Mellin[f]$ given by
\begin{equation*}
\Mellin[f](s)=\int_0^\infty f(t)\,t^{s-1}\,dt
\end{equation*}
The inverse Mellin transform of a meromorphic function $g$, denoted by $\Mellin^{-1}[g]$, reads
\begin{equation*}
\Mellin^{-1}[g](t)=\tfrac1{2\pi i}\int_{c-i\infty}^{c+i\infty}g(s)\,t^{-s}\,ds,
\end{equation*}
for some real $c$ such that the integral exists for all $t>0$.
\end{definition}
In general, the Mellin transform is defined only in some region of the complex plane. This region turns out to be a strip,
called the fundamental strip (see \cite[Definition 1]{Flajolet}). If $f(t)=\Oz(t^\alpha)$ and $f(t)=\Oinf(t^\beta)$, then $\Mellin[f](s)$ exists at least in the strip $-\alpha<\Re (s)<-\beta$ (cf. \cite{Flajolet}, Lemma 1). The invertibility of Mellin transform is addressed by the following theorem.

\begin{theorem}[\cite {Flajolet}, Theorem 2 ]
\label{thm:Mellin_inverse}
Let $f$ be a continuous function. If $c$ is a real number belonging to the fundamental strip of $\Mellin[f]$ and $\R \ni y \mapsto\Mellin[f](c+iy)$ is Lebesgue integrable, then for any $t\in\R^+$
\begin{equation*}
f(t)=\Mellin^{-1}\Big[\Mellin[f]\Big](t)=\tfrac1{2\pi i}\int_{c-i\infty}^{c+i\infty}\Mellin[f](s)\,t^{-s}\,ds.
\end{equation*}
\end{theorem}

We have established the framework and now we are ready to formulate and prove the main results.

\section{General results}\label{sec:general}

The Mellin transform has a direct application to the study of the asymptotic expansions of heat traces. Let us start with the following lemma.

\begin{lem}
\label{lem:hk_Mellin}
Let $P$ be an operator such that its zeta-function is well-defined with abscissa of convergence $L$.
Then for $\Re(s)>L$
\begin{equation}\label{MellinFund}
\Mellin[\heat_P](s)=\Gamma(s)\,\zeta_P(s).
\end{equation}
\end{lem}

\begin{proof}
For $\Re(s)>L$ we pick any $\alpha$ such that $\Re(s)>\alpha>L$. From Corollary \ref{cor:hk_order} we know that $\heat_P(t)=\Oz(t^{-\alpha})$ and the integral 
\begin{equation*}
\Mellin[\heat_P](s)=\int_0^\infty\heat_P(t)\,t^{s-1}\,dt
\end{equation*}
converges (absolutely) at 0.
It also converges absolutely at $\infty$ for any $s$, because
\[
\heat_P(t)=e^{-t\lambda_0}\big[M_0+\sum_{n=1}^\infty M_n\,e^{-t(\lambda_n-\lambda_0)}\big]=\Oinf(e^{-t\lambda_0}).
\]

As the series \eqref{hk_sum} has all its terms positive, it is absolutely convergent for $t>0$, so we can exchange the sum with integral in the following calculation
\begin{align}
\Mellin[\heat_P](s)
&=\int_0^\infty \sum_{n=0}^{\infty}M_n\,e^{-t\lambda_n}\,t^{s-1}\,dt
=\sum_{n=0}^{\infty}M_n\int_0^\infty e^{-t\lambda_n} \,t^{s-1}\,dt\notag\\
&=\sum_{n=0}^{\infty}M_n\,\lambda_n^{-s}\int_0^\infty e^{-y}\,y^{s-1}\,dy
=\zeta_P(s)\,\Gamma(s).\label{Mellin[hk]}
\end{align}
\end{proof}

For further convenience we adopt the notation
\begin{align*}
\ZZ(s) \vc \Gamma(s) \zeta_P(s),
\end{align*}
for $\Re(s) > L$.  If moreover, $\zeta_P$ extends to a meromorphic function on some larger region $D \subset \C$ then $\ZZ$ also has a meromorphic extension to $D$, since $\Gamma$ is meromorphic on $\C$. For a meromorphic function $f$ we also denote by $\PP_f(D)$ the set of its poles contained in the region $D \subset \C$.

The inverse of the relation \eqref{MellinFund} (compare \cite[Theorem 5]{Flajolet}) produces an expansion of $\heat_P$:

\begin{theorem}
\label{thm:hk_expansion}
Let $P$ be an operator such that: 
\begin{enumerate}
\item $\zeta_P$ is well-defined with abscissa of convergence $L\in\R$.
\item $\zeta_P$ has a meromorphic continuation to the half-plane $\Re (s)>L'$ for some real $L'<L$.
\item \label{assum:zeta_bound_vert} There exist real numbers $c, R$, such that $L'<-R<L<c$, and $\ZZ$ is regular and Lebesgue integrable on lines $\Re(s)=-R$ and $\Re(s)=c$.
\item \label{assum:zeta_bound_horiz} 
There exists an increasing sequence $(y_k)_{k \in \Z}$, with $y_0=0$ and $y_k\to\pm \infty$ as $k \to \pm \infty$, such that
\begin{equation}\label{assum:zeta_b}
\sup_{x\in[-R,c]}\abs{\ZZ(x + iy_k)}\to0,
\end{equation}
as $k\to\pm\infty$ and the suprema for all $k \in \Z^*$ are finite.
\end{enumerate}
Let $D_k$ denote a rectangle $\{x+iy\,|\,-R\leq x\leq c, \; y_{-k}\leq y\leq y_k\}\subset\C$ for $k\in\Np$ and $D_0=\emptyset$ and let $S_k \vc \PP_\ZZ(D_k\setminus D_{k-1})$ (see Figure \ref{f1}).

Then, for $t>0$, we have
\begin{equation}
\label{theorem_claim}
\heat_P(t)=\sum_{k=1}^\infty \sum_{s\in S_k} r_s(t)+F_R(t),
\end{equation}
where the, possibly infinite, series over $k$ is convergent with
\begin{align*}
r_s(t)&\vc\Rez{s'=s}\left(\Gamma(s')\,\zeta_P(s')\,t^{-s'}\right),\\
F_R(t)&\vc\tfrac1{2\pi i}\int_{-R-i\infty}^{-R+i\infty}\Gamma(s)\,\zeta_P(s)\,t^{-s}\,ds.
\end{align*}
Moreover, $F_R(t) = \oz(t^R)$.
\end{theorem}

\begin{figure}[t]
\begin{center}
\begin{picture}(300,300)
	\put(0,0){\includegraphics[scale=1]{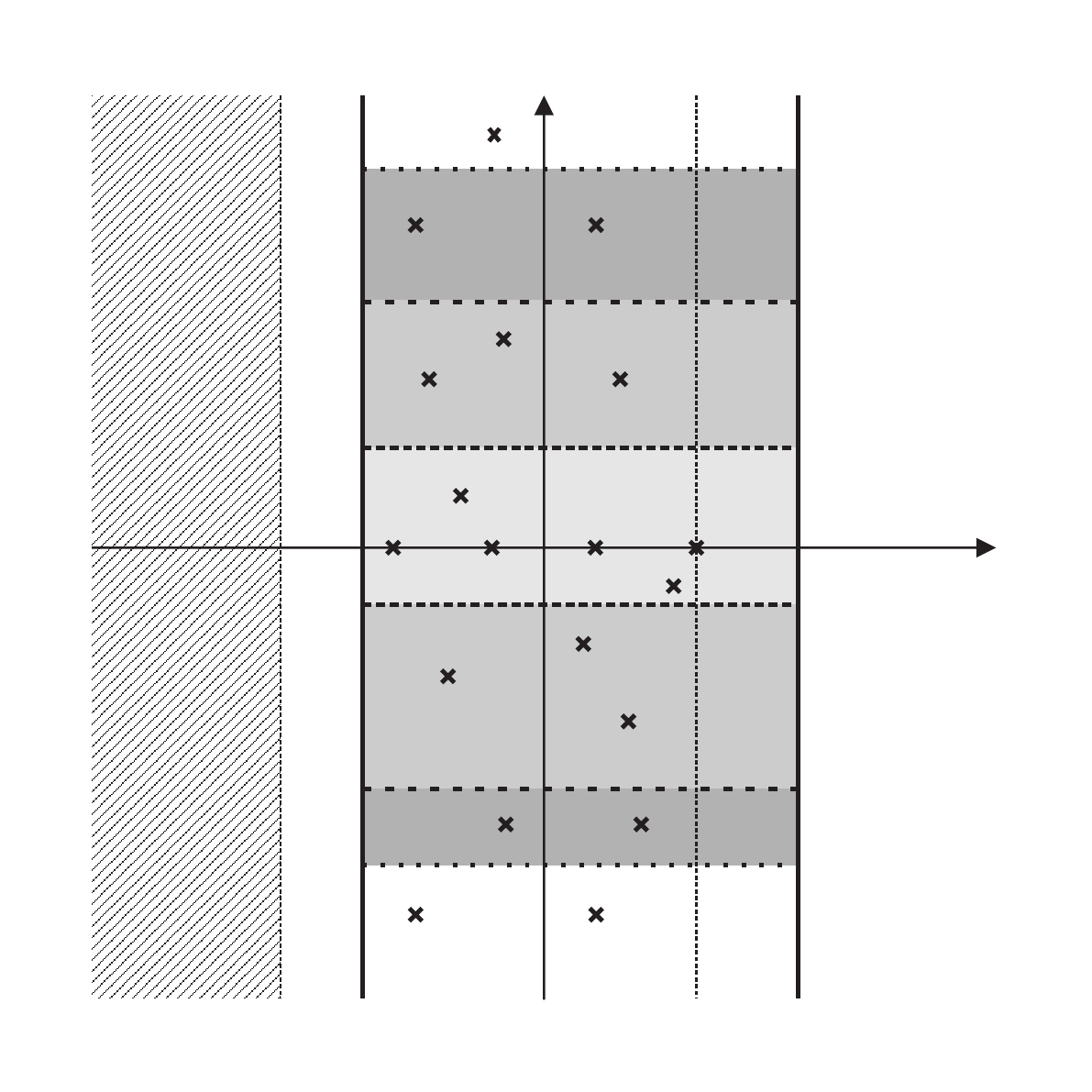}}
  	\put(300,150){$\Re(s)$}
  	\put(160,320){$\Im(s)$}
  	\put(80,10){$L'$}
  	\put(105,10){$-R$}
  	\put(217,10){$L$}
  	\put(248,10){$c$}
  	\put(260,67){$y_{-3}$}
  	\put(260,93){$y_{-2}$}
  	\put(260,152){$y_{-1}$}
  	\put(260,200){$y_{1}$}
  	\put(260,247){$y_{2}$}
  	\put(260,289){$y_{3}$}
\end{picture}
\caption{Illustration of Theorem \ref{thm:hk_expansion}. The crosses stand for poles of the function $\ZZ$. \label{f1}}
\end{center}
\end{figure}
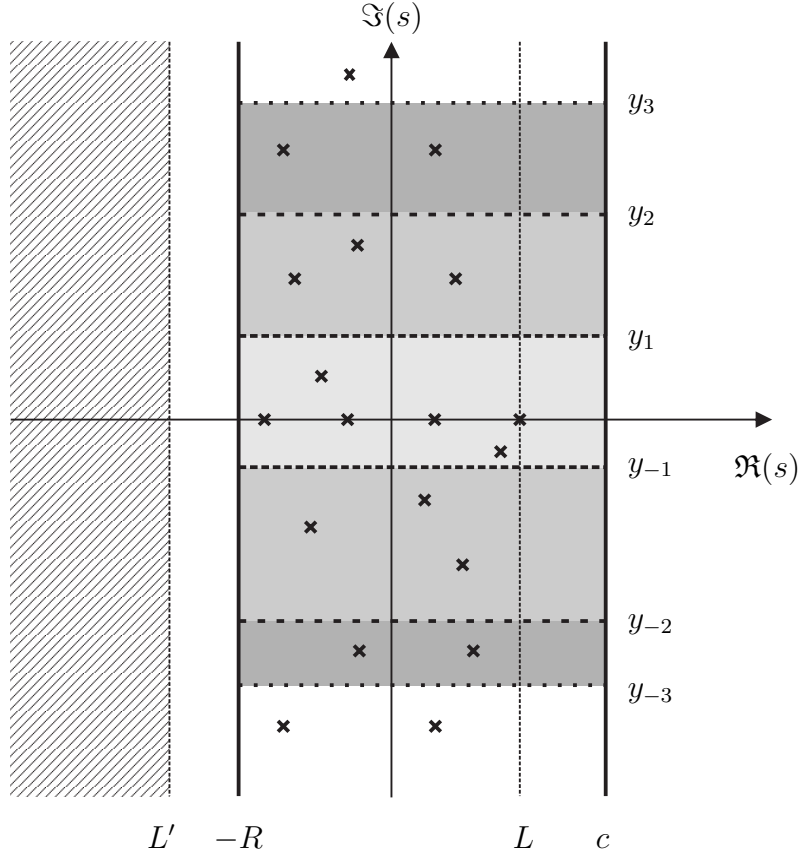

\begin{proof}
On the strength of assumptions \ref{assum:zeta_bound_vert} and \ref{assum:zeta_bound_horiz} the function $\ZZ$ is regular at the boundary of $D_k$. Thus, by the residue theorem,
\begin{equation}
\label{finite_contour}
\tfrac1{2\pi i}\int_{\partial D_k}\ZZ(s)\,t^{-s}\,ds
=\sum_{s\in \PP_\ZZ(D_k)}r_s(t)
=\sum_{m=1}^k\sum_{s\in S_m} r_s(t) ,
\end{equation}
where the contour $\partial D_k$ is oriented counter-clockwise. In the sum above only a finite number of residues is taken into account, as the region $D_k$ is bounded and the set of poles $\PP_\ZZ(D_k)$ has no accumulation points. 
 Let us decompose the boundary of the rectangle into four sides
\begin{gather*}
I_{-R}(k)\vc\int_{-R+i y_{-k}}^{-R+i y_k}\ZZ(s)\,t^{-s}\,ds,\qquad 
I_c(k)\vc\int_{c+i y_{-k}}^{c+i y_k}\ZZ(s)\,t^{-s}\,ds,\\
I_H^\pm(k)\vc \int_{-R+i y_{\pm k}}^{c+i y_{\pm k}}\ZZ(s)\,t^{-s}\,ds,
\end{gather*}
so that
\begin{equation*}
\int_{\partial D_k}\ZZ(s)\,t^{-s}\,ds=I_{c}(k)-I_H^+(k)-I_{-R}(k)+I_H^-(k).
\end{equation*}

We now estimate 
\begin{equation}
\label{horiz_integral_estimate}
\abs{I_H^\pm(k)}
=\abs{\int_{-R}^c\ZZ(x+ iy_{\pm k})\,t^{-x - iy_{\pm k}}\,dx}
\leq \sup_{x\in[-R,c]}\abs{\ZZ(x + iy_{\pm k})}\,\int_{-R}^ct^{-x}\,dx\xrightarrow{k\to\infty}0.
\end{equation}

To analyse the integrals $I_{-R},I_c$ we will use the assumption \ref{assum:zeta_bound_vert}. Integrability over the line $\Re (s)=c$ allows us to apply Theorem \ref{thm:Mellin_inverse}, which, together with Lemma \ref{lem:hk_Mellin}, gives us the following limit
\begin{equation}
\label{lim_Ic}
\tfrac1{2\pi i}\,\lim_{k\to\infty}I_c(k)=\heat_P(t).
\end{equation}
On the other hand, integrability of $\ZZ$ along the line $\Re (s)=-R$ allows us to write
\begin{equation*}
\tfrac1{2\pi i}\,\lim_{k\to\infty}I_{-R}(k)=F_R(t).
\end{equation*}
Thus, by taking the limit $k\to\infty$ of equation \eqref{finite_contour} we get
\begin{equation*}
\sum_{m=0}^\infty \sum_{s\in S_m} r_s(t)=\lim_{k\to\infty}\tfrac1{2\pi i}\int_{\partial D_k}\ZZ(s)t^{-s}ds=\heat_P(t)-F_R(t).
\end{equation*}
To finish the proof we observe that
\begin{multline*}
F_R(t)/t^R
=\tfrac1{2\pi}\int_{-\infty}^\infty \ZZ(-R+iy)\,t^{-iy}\,dy
=\tfrac1{2\pi}\Fourier[y\mapsto \ZZ(-R+iy)] \left(-\tfrac{\log t}{2\pi} \right)\xrightarrow{t\to0}0,
\end{multline*}
where $\Fourier$ denotes the Fourier transform, and the limit is a consequence of the Riemann--Lebesgue lemma.
\end{proof}

The series in \eqref{theorem_claim} is just the sum of residues of the function $\ZZ(s)t^{-s}$ over the poles contained in the strip $-R<\Re(s)<c$ and one could be tempted to write it down as
\begin{equation}
\label{residue_sum}
\sum_{s\in S} r_s(t),
\end{equation} 
where summation goes over $S=\PP_\ZZ(\{s\in\C:-R<\Re(s)<c\})$. Whilst this second form looks simpler, it ignores the information about grouping and arrangement of terms, which may be significant. Indeed, Theorem \ref{thm:hk_expansion} states that the series over $k$ in formula \eqref{theorem_claim} is convergent, and in general this may only be a conditional convergence. That is why residues have to be grouped into (finite) sums over $S_k$, and then summed in the order given by index $k$. The grouping and order are consequences of the choice of the sequence $(y_k)$ and the assumption \eqref{assum:zeta_b} may fail for a different choice of sequence.

However, if the function $\ZZ$ has only a finite number of residues in the strip $-R<\Re(s)<c$ or the sum \eqref{residue_sum} is absolutely convergent then rearrangements of terms are allowed and one can safely write
\begin{equation}
\label{theorem_claim_simplified}
\heat_P(t)=\sum_{s\in S} r_s(t) +F_R(t),
\end{equation}
instead of \eqref{theorem_claim}. One clue about the absolute convergence is given by the following result:
\begin{prop}\label{rem:absolute_conv}
If the hypothesis of Theorem \ref{thm:hk_expansion} is fulfilled with the assumption \eqref{assum:zeta_b} altered for a stronger one:
\begin{equation}
\label{hypo_absolute}
\sum_{k\in\Z^*}\sup_{x\in[-R,c]}\abs{\ZZ(x+iy_k)}<\infty,
\end{equation}
then for any $t>0$
\begin{equation*}
\sum_{k=1}^\infty \abs{\sum_{s\in S_k} r_s(t)}<\infty.
\end{equation*}
\end{prop}
\begin{proof}
First note that for any $t>0$
\begin{align*}
\sum_{k=1}^\infty \abs{\sum_{s\in S_k} r_s(t)}&=\tfrac1{2\pi i}\sum_{k=1}^\infty \abs{\int_{\partial (D_k\setminus D_{k-1})}\ZZ(s)\,t^{-s}\,ds}.
\end{align*}
Now, for any $k \geq 1$ we decompose the boundary $\partial (D_k\setminus D_{k-1})$ as in the proof of Theorem \ref{thm:hk_expansion} and estimate
\begin{multline*}
\abs{\int_{\partial (D_k\setminus D_{k-1})}\ZZ(s)\,t^{-s}\,ds} \leq \abs{I_{-R}(k) - I_{-R}(k-1)} + \abs{I_{c}(k)-I_{c}(k-1)} + \\
+ \abs{I_{H}^+(k-1)} + \abs{I_{H}^-(k-1)} + \abs{I_{H}^+(k)} + \abs{I_{H}^-(k)}, 
\end{multline*}
with the convention $I_H^{\pm}(0) \vc 0$ to comply with $D_0 = \emptyset$. So the contributions of the horizontal integrals $I_H^{\pm}(k)$, $I_H^{\pm}(k+1)$ add together instead of canceling out as it happened in the proof of Theorem \ref{thm:hk_expansion}.

Therefore,
\begin{multline*}
\sum_{k=1}^\infty \abs{\int_{\partial (D_k\setminus D_{k-1})}\ZZ(s)\,t^{-s}\,ds} \leq \int_{-R - i \infty}^{-R+i \infty} \abs{\ZZ(s) t^{-s}} ds + \\
+ \int_{c - i \infty}^{c+i \infty} \abs{\ZZ(s) t^{-s}} ds + 2 \sum_{k=1}^\infty \left( \abs{I_{H}^+(k)} + \abs{I_{H}^-(k)} \right). 
\end{multline*}

The convergence of the integrals along the vertical lines follows since $\ZZ$ is Lebesgue integrable (assumption \ref{assum:zeta_bound_vert} of Theorem \ref{thm:hk_expansion}).

On the other hand, the sum over horizontal contributions can be estimated as in \eqref{horiz_integral_estimate}:
\begin{align*}
\sum_{k=1}^\infty \left(\abs{I_H^+(k)} + \abs{I_H^-(k)}\right) \leq \sum_{k \in \Z^*} \sup_{x\in[-R,c]}\abs{\ZZ(x + iy_{\pm k})}\,\int_{-R}^ct^{-x}\,dx,
\end{align*}
which is finite by assumption \eqref{hypo_absolute}.
\end{proof}

In most cases (see Section \ref{sec:examples}), one can also avoid the grouping of the residues into $S_k$'s, by finding some denser sequence $(y_k)_{k\in\Z}$ such that each $S_k$ contains only one pole of $\ZZ$. Also, if there are poles of $\ZZ$ lying on a common line $\Im(s) = \text{const.}$, one can resort to the more general Theorem \ref{thm:hk_asymptotic_general}.  However, for a denser sequence the assumptions \eqref{assum:zeta_b} or \eqref{hypo_absolute} of the Theorem \ref{thm:hk_expansion} may not be fulfilled.

Unfortunately, we were not able to tailor an example, where the analytic structure of $\zP$ is such that the grouping or arrangement of the terms are important. It might turn out that the operatorial aspect of heat traces, which leads to a specific subclass of general Dirichlet series, implies that one can always replace formula \eqref{theorem_claim} with \eqref{theorem_claim_simplified}. However, as the problem is open, we emphasise once again that in general formulae the series in \eqref{theorem_claim} is conditionally convergent only. 

It is instructive to write down explicitly an individual term $r_s(t)$. If the function $\ZZ$ has a pole of order $n$ at $s_0$, then it has a Laurent expansion $\ZZ(s) = \sum_{k=-n}^\infty b_k(s_0) \,(s-s_0)^k$ in some open punctured disc with the center at $s_0$. On the other hand,
\begin{align*}
t^{-s}=e^{-s_0\log t}e^{-(s-s_0)\log t}=t^{-s_0}\sum_{k=0}^\infty \frac{(-\log t)^k}{k!}(s-s_0)^k, && \forall \; s_0, \, s \in \C, \; t >0.
\end{align*}
Therefore, the residue $r_{s_0}(t)$ being the coefficient of $(s-s_0)^{-1}$ in the Laurent expansion of $\ZZ(s)t^{-s}$ at $s=s_0$ reads
\begin{align}
\label{r_s(t)_asymptotics}
r_{s_0}(t)=t^{-s_0}\sum_{k=0}^{n-1}\frac{b_{-k-1}(s_0)}{k!}(-\log t)^k.
\end{align}
Note that $r_{s_0}(t)=\Oz(t^{-\Re (s_0)}(\log t)^{n-1})$, what also means that $r_{s_0}(t)=\Oz(t^{-\Re (s_0)+\delta})$ for arbitrarily small $\delta>0$. 

\begin{remark}
One could in principle allow the function $\ZZ$ to have \textit{essential singularities} as long as they are isolated (see \cite[last point on p. 453]{CPRS} for a motivation). For such functions the residue (i.e. the $-1^{\text{st}}$ term of the Laurent expansion) is well-defined. Since $\Gamma$ is a meromorphic function and $s \mapsto t^{-s}$ is an entire one for all $t>0$, the function
\begin{align}\label{UDt}
s \mapsto \ZZ(s) t^{-s},
\end{align}
has isolated singularities only. At an essential singularity, the formula \eqref{r_s(t)_asymptotics} yields an infinite series. This series is absolutely convergent for every $t>0$ in the punctured disc of convergence of the Laurent expansion of $\ZZ$ at $s_0$.  Indeed, in the interior of this punctured disc the Laurent series of the function \eqref{UDt} is a product of two absolutely convergent series, and as such it is absolutely convergent. 
\end{remark}

\subsection{Asymptotic expansions}\label{Asymptotic}
Theorem \ref{thm:hk_expansion} gives us information about the behavior of $\heat_P(t)$ at $t=0$ up to the order $t^R$. If $\zeta_P$ can be meromorphically continued to the whole complex plane and satisfies suitable growth conditions, then Theorem \ref{thm:hk_expansion} can give us the behavior of $\heat_P$ at 0 up to an arbitrary finite order, i.e. an asymptotic expansion. Recall the definition \cite{Copson, Hardy_div, Erdelyi}:

\begin{definition}\label{DefAsymp}
Let $(\phi_n)_{n\in\N}$ be a sequence of functions from $\R^+$ to $\C$. We call this sequence an asymptotic scale at $t=0$ if for any $n\in\N$ we have $\phi_{n+1}(t)=\oz(\phi_n(t))$.

For a function $f:\R^+\to\C$ the formal series $\sum_{n=0}^\infty \phi_n(t)$, with $(\phi_n)_{n\in\N}$ being an asymptotic scale, is called an \emph{asymptotic expansion} (or \emph{asymptotic series}) of $f$ at $t=0$ if for any $N\in\N$
\begin{equation*}
f(t)-\sum_{n=0}^N \phi_n(t)=\Oz(\phi_{N+1}(t)).
\end{equation*}
In this case, we write
\begin{align*}
f(t) \tzero \sum_{n=0}^\infty \phi_n(t).
\end{align*}
\end{definition}

Now by an iterative argument exploiting Theorem \ref{thm:hk_expansion} we obtain the asymptotic expansion of $\heat_P$, which is the main result of the paper:

\begin{theorem}\label{thm:hk_asymptotic_general}
Let $P$ be a positive operator with compact inverse such that:
\begin{enumerate}
\item $\zeta_P$ is well-defined with abscissa of convergence $L$.
\item $\zeta_P$ has a meromorphic continuation to the whole complex plane.
\item \label{assum3:general} There exists a sequence $(R_n)_{n\in\N}$ of real numbers strictly increasing to infinity, such that $-R_0>L$, $-R_n<L$ for $n\geq1$, and for each $n\in\N$ function $\ZZ(s)=\Gamma(s)\zeta_P(s)$ is regular and Lebesgue integrable over the vertical line $\Re(s)=-R_n$.
\item \label{assum4:general} For each $n\geq1$ there exists a strictly increasing sequence $(y^{(n)}_k)_{k\in\Z}$ with $y^{(n)}_0=0$ and $y^{(n)}_k\to\pm \infty$ as $k \to \pm \infty$ such that
\begin{equation*}
\sup_{x\in[-R_n,-R_{n-1}]}\abs{\ZZ(x+iy^{(n)}_k)}\to0,
\end{equation*}
as $k\to\pm\infty$ and the suprema for all $k \in \Z^*$ are finite.
\end{enumerate} 
For $n,k\geq1$ let $D_k^n$ denote a rectangle $\{x+iy\,|\,-R_n\leq x\leq -R_{n-1}, \; y_{-k}^{(n)}\leq y\leq y_k^{(n)}\}\subset\C$ and $D_0^n=\emptyset$ for $n \geq 1$. Set $S_k^n = \PP_\ZZ(D_{k}^n\setminus D_{k-1}^n)$ (see Figures \ref{f1} and \ref{f2}).

Then, for $t>0$
\begin{equation}
\label{hk_asymptotic_general}
\heat_P(t) \tzero \sum_{n=1}^{\infty} \sum_{k=1}^\infty\sum_{s\in S^n_k}r_s(t),
\end{equation}
where
\begin{equation*}
r_s(t)=\Rez{s'=s}\left(\ZZ(s')\,t^{-s'}\right).
\end{equation*}
\end{theorem}

The RHS of \eqref{hk_asymptotic_general} is to be understood as the asymptotic (formal) series $\sum_{n=1}^\infty\phi_n(t)$, with $\phi_n:\R^+\to\C$ being an asymptotic scale defined by the convergent series
\begin{equation}
\label{phi_n_scale}
\phi_n(t) \vc\sum_{k=0}^\infty\sum_{s\in S^n_k}r_s(t),\qquad \text{for }n\geq1.
\end{equation}

\begin{figure}[t]
\begin{center}
\begin{picture}(300,300)
	\put(0,0){\includegraphics[scale=1]{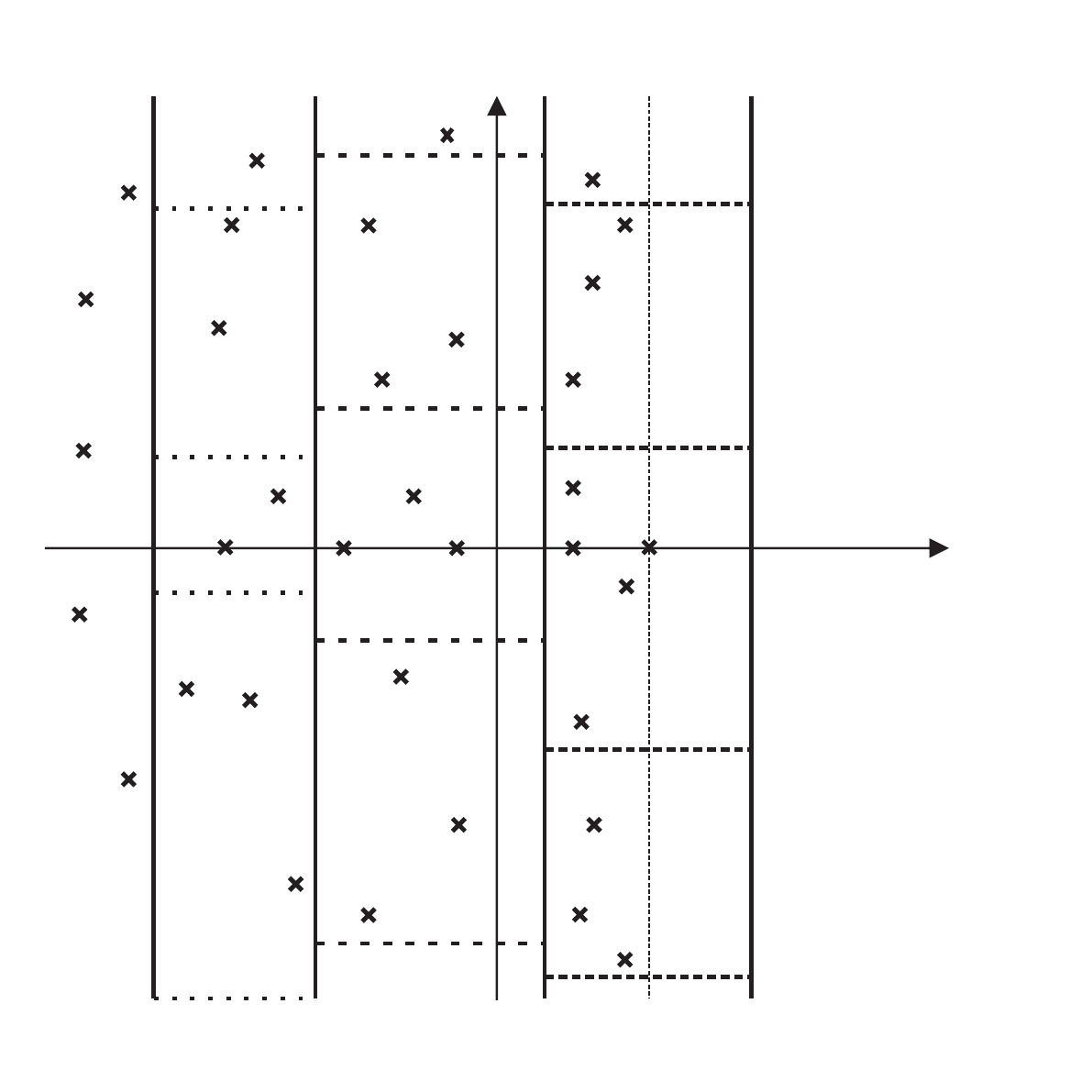}}
  	\put(300,150){$\Re(s)$}
  	\put(143,320){$\Im(s)$}
  	\put(33,10){$-R_3$}
  	\put(85,10){$-R_2$}
  	\put(158,10){$-R_1$}
  	\put(200,10){$L$}
  	\put(225,10){$-R_0$}
  	\put(242,33){$y_{-2}^{(1)}$}
  	\put(242,105){$y_{-1}^{(1)}$}
  	\put(242,200){$y_{1}^{(1)}$}
  	\put(242,277){$y_{2}^{(1)}$}
  	\put(80,42){$y_{-2}^{(2)}$}
  	\put(80,138){$y_{-1}^{(2)}$}
  	\put(80,212){$y_{1}^{(2)}$}
 	\put(80,300){$y_{2}^{(2)}$}
 	\put(30,26){$y_{-2}^{(3)}$}
 	\put(30,152){$y_{-1}^{(3)}$}
 	\put(30,197){$y_{1}^{(3)}$}
  	\put(30,265){$y_{2}^{(3)}$}
\end{picture}
\caption{Illustration of Theorem \ref{thm:hk_asymptotic_general}. See also Figure \ref{f1}.\label{f2}}
\end{center}
\end{figure}
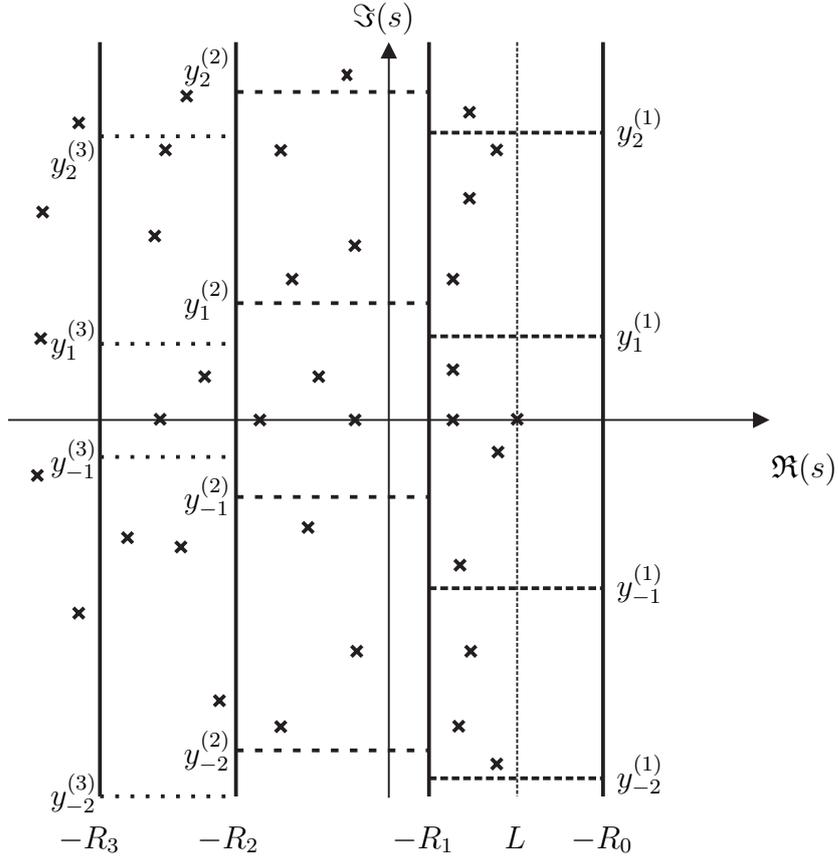 

\begin{proof}
Let us first apply Theorem \ref{thm:hk_expansion} to $\zeta_P$ with $c=-R_0, R=R_1$ and $y_k=y^{(1)}_k$ for $k\in\Z$. All of the assumptions are readily fulfilled and we obtain
\begin{equation}
\label{heat_phi_1}
\heat_P(t)=\phi_1(t)+F_{R_1}(t),
\end{equation}
where $\phi_1$ is given by \eqref{phi_n_scale}.

Then, we use Theorem \ref{thm:hk_expansion} again with $c=-R_1, R=R_2$ and $y_k=y^{(2)}_k$ for $k\in\Z$. Strictly speaking in this case $c$ is not greater than $L$ as required by assumption \ref{assum:zeta_bound_vert}, but as long as $\ZZ$ is regular and Lebesgue integrable on the lines $\Re(s)=-R_1$ and $\Re(s)=-R_0$ the same arguments apply. The contour integral, now being $2\pi i(F_{R_1}(t)-F_{R_2}(t))$, is again equal to the sum of residues, giving
\begin{equation*}
F_{R_1}(t)=\phi_2(t)+F_{R_2}(t)
\end{equation*}
instead of \eqref{heat_phi_1}. Repeating this argument for the function $\zeta_P$ in each strip $-R_{n+1}\leq\Re (s)\leq-R_n$ we obtain a recurrence relation:
\begin{equation}
\label{phi_recurrence}
F_{R_n}(t)=\phi_{n+1}(t)+F_{R_{n+1}}(t).
\end{equation}
Thus, for any $N \in \Np$ we have
\begin{equation*}
\heat_P(t)=\sum_{n=1}^N\phi_n(t)+F_{R_N}(t),
\end{equation*}
where $F_{R_N}(t)=\oz(t^{R_N})$ and by Definition \ref{DefAsymp} of an asymptotic expansion we conclude that
\begin{align*}
\heat_P(t) \tzero \sum_{n=1}^\infty \phi_n(t).
\end{align*}
\end{proof}

The choice of the sequence $(R_n)$ determines the asymptotic scale $(\phi_n)$ of the expansion \eqref{hk_asymptotic_general}. Thus, in our method the residues are first summed in each vertical strip $\{s \in \C \; : \; -R_{n+1} < \Re(s) < -R_{n} \}$ yielding $\phi_n(t)$, and then the contributions from subsequent strips form an asymptotic series.

\subsection{Exact and almost exact expansions}

Having an asymptotic series for $\heat_P(t)$, it is natural to ask whether this series converges for some $t$. This can be checked by analysing the limit of $F_{R_N}(t)$ as $N\to\infty$. If this limit exists and $F_{R_N}\to F_\infty < \infty$ locally uniformly on the interval $(0,T)$ for some $T>0$, then the series in \eqref{hk_asymptotic_general} converges locally uniformly on $(0,T)$ to $\heat_P(t)-F_\infty(t)$. Note that $F_\infty$ is $\oz(t^R)$ for every $R\in\R$. Let us adopt the following definition:

\begin{defn}\label{def:almost}
If the series on the RHS of \eqref{hk_asymptotic_general} converges (to $\heat_P(t)-F_\infty(t)$) locally uniformly on $(0,T)$ for some $T>0$ then we say that $\heat_P$ has an \emph{almost exact expansion} on $(0,T)$. If moreover $F_\infty = 0 $ identically, then $\heat_P$ has an \emph{exact expansion} on $(0,T)$.
\end{defn}

In Section \ref{sec:examples} we will meet various examples of almost exact and exact expansions of heat traces, with $T < \infty$ and $T = \infty$ as well. Let us now give a general sufficient condition for the existence of an exact expansion.

\begin{theorem}\label{cor:hk_exact}
Let the assumptions of Theorem \ref{thm:hk_asymptotic_general} be fulfilled and let the estimate 
\begin{equation}\label{exact}
\abs{\ZZ(-R_n+iy)}\leq C_ne^{-\epsilon_n\abs{y}}
\end{equation}
hold for every $y\in\R$ and $n\in\N$, where $C_n, \epsilon_n$ are some positive constants for $n\in\N$. Assume moreover that the sequence $\sqrt[R_n]{C_n/\epsilon_n}$ is bounded for $n\in\N$. 

Then,
\begin{equation}
\label{hk_exact_general}
\heat_P(t)=\sum_{n=1}^\infty\sum_{k=0}^\infty\sum_{s\in S^n_k}r_s(t),
\end{equation}
for $t\in(0,T)$, where 
\begin{equation*}
T\vc\left(\limsup_{n\to\infty}\sqrt[R_n]{\frac{C_n}{\epsilon_n}}\right)^{-1}
\end{equation*}
and the series over $n$ is locally uniformly convergent on $(0,T)$.
\end{theorem} 

\begin{proof}
Let us estimate the reminder $F_{R_n}(t)$ as follows
\begin{multline*}
\abs{F_{R_n}(t)}=\tfrac1{2\pi}\abs{\int_{-\infty}^\infty \ZZ(-R_n+iy)t^{R_n-iy}dy}
\leq\tfrac1{2\pi}\int_{-\infty}^\infty C_ne^{-\epsilon_n\abs{y}}t^{R_n}dy\\
=\tfrac{C_nt^{R_n}}{2\pi}2\int_0^\infty e^{-\epsilon_ny}dy
=\tfrac{C_nt^{R_n}}{2\pi}\tfrac2{\epsilon_n}.
\end{multline*}

Let $0<T'<T$. Then, for any $t\in(0,T']$
\begin{equation*}
\limsup_{n\to\infty}\left(t\sqrt[R_n]{\tfrac{C_n}{\epsilon_n}}\right)=\tfrac tT\leq\tfrac{T'}T<1.
\end{equation*}
Hence, for sufficiently large $n$ we have $t\sqrt[R_n]{\frac{C_n}{\epsilon_n}}<a$, where $a\in(T'/T,1)$ is some constant independent of $t$. Then
\begin{equation}
\label{F_Rn_bound}
\abs{F_{R_n}(t)}\leq\tfrac{C_nt^{R_n}}{\epsilon_n\pi}<\tfrac{a^{R_n}}\pi\to0
\end{equation}
so $F_{R_n}(t)$ tends to 0 uniformly for $t\in(0,T']$ as $n\to\infty$. Since $T'$ can be any number in $(0,T)$, the theorem is proven.
\end{proof}

Note that again the order of summation in \eqref{hk_exact_general} is important and the convergence may be only conditional. As in the case of the vertical sum over $k$ (see page \pageref{residue_sum}), the convergence properties depend on the choice of the $(R_n)$ sequence. And as in Proposition \ref{rem:absolute_conv} we can refine the assumptions of Theorem \ref{cor:hk_exact} to obtain an absolute convergence of the series over $n$ in formula \eqref{hk_exact_general}.

Let us first adopt the following definition:

\begin{Def}\label{def:absolutely_exact}
Let the operator $P$ be such that
\begin{align}\label{exact_series}
\heat_P(t) = \sum_{n=1}^{\infty} \phi_n(t), && \text{for } t \in (0,T),
\end{align}
with some $0<T\leq + \infty$, where $(\phi_n)_{n\in\N}$ is the asymptotic scale given by \eqref{phi_n_scale}. We say that the heat trace associated with $P$ has an \emph{absolutely exact} expansion on $(0,\wt{T})$ if the series
\begin{align*}
\sum_{n=1}^{\infty} \abs{\phi_n(t)}
\end{align*}
is locally uniformly convergent on $(0,\wt{T})$ with some $0<\wt{T} \leq T$.
\end{Def}

Note, that since the RHS of \eqref{exact_series} is not in general a Taylor series, its domain of convergence does not necessarily coincide with that of absolute convergence, hence $\wt{T} \leq T$. In particular if $0 = \wt{T} < T$, then the expansion of the heat trace associated with $P$ will be exact on $(0,T)$, but nowhere absolutely exact.

By using a similar reasoning as the one used in the proof of Proposition \ref{rem:absolute_conv}, we can refine the assumptions of Theorem \ref{cor:hk_exact} to obtain an absolutely exact heat trace expansion.

\begin{prop}\label{prop:abs}
Let the assumptions of Theorem \ref{cor:hk_exact} be fulfilled and moreover let $\log n =\oinf(R_n)$ (i.e. $R_n$ grow faster than $\log n$). Then the expansion is absolutely exact on the whole domain $(0,T)$.
\end{prop}
\begin{proof}
As announced, we proceed similarly to the proof of Proposition \ref{rem:absolute_conv}, but now we shall add together the contributions of subsequent vertical integrals. From \eqref{phi_recurrence} we have
\begin{align*}
\abs{\phi_n(t)} =\abs{F_{R_n}(t) - F_{R_{n-1}}(t)}\leq \abs{F_{R_n}(t)} + \abs{F_{R_{n-1}}(t)}.
\end{align*}
Now, \eqref{F_Rn_bound} implies
\begin{align*}
\sum_{n = 1}^{\infty} \abs{\phi_n(t)} 
\leq 2\sum_{n = 1}^{\infty} \abs{F_{R_n}(t)}
\leq 2\sum_{n = 1}^{\infty}\frac{C_n t^{R_n}}{\epsilon_n \pi} 
< \tfrac{2}{\pi}\sum_{n = 1}^{\infty}a^{R_n}.
\end{align*}
with $a \in (T'/T,1)$ for any $T' \in (0,T)$. Therefore, it suffices to show that the last series is convergent for any $a<1$. Taking $x=-\log a$ we see that it is again a general Dirichlet series \eqref{gen_Dirichlet_series} with $a_n=1, b_n=R_n$ for $n\in\N$. By Theorem \ref{thm:abscissa} we conclude that its abscissa of convergence is 
\begin{equation*}
\limsup_{n\to\infty} \frac{\log n}{R_n},
\end{equation*}
which equals $0$ by hypothesis. Thus, the series is convergent for $x>0$, i.e $a<1$.
\end{proof}

The characterisation of an absolutely exact expansion through Proposition \ref{prop:abs} is particularly useful. If we have an exact expansion of a heat trace $\heat_P$ on $(0,T)$, it is sufficient to check whether the sequence $(R_n)$ of our choice grows faster than $\log n$. In fact, all of the examples of exact expansions presented in Section \ref{sec:examples} are actually absolutely exact in the same domain.

Let us also note, that if we have an exact expansion of the heat trace for an open interval $(0,T)$, then $\heat_P$ actually provides an analytic continuation of the RHS of \eqref{hk_exact_general} to the whole half line $\R^+$.

Let us now turn to the case of an almost exact expansion of heat traces. The situation is somewhat different than that of an exact expansion, as one would need precise formulae rather than estimates to guarantee that the limit $F_{R_N}$ as $N\to\infty$ is finite, but non-zero. It may be seen as a kind of ``critical'' case, in the sense that a slight perturbation of the zeta-function renders the expansion divergent (see Section \ref{divergent}). One of these specific cases is captured by the following proposition.

\begin{Prop}\label{prop:finite}
If the operator $P$ fulfilling the assumptions of Theorem \ref{thm:hk_asymptotic_general} is such that $\PP_\ZZ(\C)$ is a finite set, then the expansion
\begin{align}\label{heat_finite}
\heat_P(t) = \sum_{s \in \PP_\ZZ(\C)} r_s(t) + F_{\infty}(t),
\end{align}
is almost exact for all $t>0$, but not exact.
\end{Prop}
\begin{proof}
Let us note that $\PP_\ZZ(\C)$ being finite requires $\PP_{\zP}(\C)$ to be finite, but also that the zeros of $\zP$ cancel all but a finite number of poles of $\Gamma$, i.e. $\zP(-n)=0$ for all but a finite number of $n\in\N$. In this case, Theorem \ref{thm:hk_asymptotic_general} yields an asymptotic expansion, which has only a finite number of terms, hence it converges for all $t>0$. As $t$ tends to infinity we have $\heat_P(t) = \OO_{t \to \infty} (e^{-t\lambda_0})$, since
\begin{align*}
\Tr\,e^{-t\,P} = \sum_{n=0}^\infty M_n e^{-t\,\lambda_n} = e^{-t \, \lambda_0} \sum_{n=0}^\infty e^{-t\,(\lambda_n-\lambda_0)} \leq e^{-t\,\lambda_0} c \text{ for } t \geq 1, 
\end{align*} 
where $c = \sum_{n=0}^\infty e^{-(\lambda_n-\lambda_0)}=e^{\lambda_0}\Tr e^{-P}<\infty$. But this is not compatible with the behaviour $t^{-s} (\log t)^k$ (compare \eqref{r_s(t)_asymptotics}) of the summands of the first term on the RHS of \eqref{heat_finite}. Hence, $F_\infty(t)$ cannot be 0 and the expansion \eqref{heat_finite} is almost exact, but not exact.
\end{proof}

\subsection{Truncated zeta-function}
\label{sec:truncated}
For further purposes it is convenient to define also the following spectral function.
\begin{definition}
For any $N \in \Np$, a \emph{truncated zeta-function} associated with the operator $P$ is a complex function
\begin{equation}\label{zT}
\C \supset \Dom(\zP) \ni s \mapsto \zT(s)=\sum_{n=N}^{\infty} M_n\,\lambda_n^{-s}.
\end{equation}
\end{definition}

The two zeta-functions are related by
\begin{align}
\label{zT_to_zP}
\zP(s) = \sum_{n=0}^{N-1} M_n\,\lambda_n^{-s} + \zT(s).
\end{align}
The first term of the RHS of the above formula is an entire function of $s$, hence the analytic properties of $\zP$ and $\zT$ are identical.

We can actually relax slightly the growth rate assumptions \ref{assum:zeta_bound_vert} and \ref{assum:zeta_bound_horiz} of Theorem \ref{thm:hk_expansion} by considering the truncated zeta-function \eqref{zT} at the place of $\zP$. In Section \ref{sec:examples} we will encounter explicit examples, when $\zP$ grows too fast on vertical lines, but nevertheless the inverse Mellin transform technique can be applied by using $\zT$ instead.

\begin{Prop}\label{prop:truncated}
If the assumptions of Theorem \ref{thm:hk_expansion} hold for $\zT$ and $\trunc{\ZZ} : s \mapsto \Gamma(s)\zT(s)$ with some finite $N \in \Np$, then for $t>0$ we have
\begin{equation*}
\heat_P(t)=\sum_{n=0}^{N-1} M_n e^{-t \lambda_n} + \sum_{n=0}^\infty \sum_{s\in S_n} \trunc{r}_s(t)+\trunc{F}_R(t),
\end{equation*}
where
\begin{align*}
\trunc{r}_s(t)&\vc\Rez{s'=s}\left(\trunc{\ZZ}(s')\,t^{-s'}\right),\\
\trunc{F}_R(t)&\vc\tfrac1{2\pi i}\int_{-R-i\infty}^{-R+i\infty}\trunc{\ZZ}(s)\,t^{-s}\,ds = \oz(t^R).
\end{align*}
\end{Prop}
\begin{proof}
The proof of Theorem \ref{thm:hk_expansion} will work equally well for this case with the only difference at line \eqref{lim_Ic} which will now read
\begin{equation*}
\tfrac1{2\pi i}\,\lim_{n\to\infty}\trunc{I}_c(n)=\Mellin^{-1}[\trunc{\ZZ}](t).
\end{equation*}
Applying the inverse Mellin transform to relation \eqref{zT_to_zP} multiplied by $\Gamma(s)$ we get (using Lemma \ref{lem:hk_Mellin})
\begin{equation*}
\heat_P(t)=\sum_{n=0}^{N-1} M_n \Mellin^{-1}[\lambda_n^{-s}\Gamma(s)] + \Mellin^{-1}[\trunc{\ZZ}](t).
\end{equation*}
Note that for any $n\in\N$ the function $\lambda_n^{-s}\Gamma(s)$ is integrable over any vertical line $\Re(s)=c\notin\N$ (as $\Gamma$ decays exponentially on verticals, compare \eqref{gamma_vert}) and that it equals to $\Mellin[\exp(-t\lambda_n)](s)$ (cf. \eqref{Mellin[hk]}). The claim then follows from Theorem \ref{thm:Mellin_inverse}.
\end{proof}

As $\zP$ and $\zT$ differ by an entire function, the assumptions \textit{(i)} and \textit{(ii)} of Theorem \ref{thm:hk_expansion} are fulfilled by $\zT$ if and only if they are fulfilled by $\zP$. In view of Proposition \ref{prop:truncated} it is convenient to introduce the notion of a truncated heat trace:
\begin{equation}
\label{truncated_htr}
\trunc{\heat}_P(t)\vc\heat_P(t)-\sum_{n=0}^{N-1} M_n e^{-t \lambda_n}.
\end{equation}
Then all of the considerations about the truncated zeta-function can be expressed in one simple statement:

\begin{remark}\label{remark_trunc}
All of the results concerning the asymptotic and exact expansions of heat traces hold with $\zP$ altered for the truncated zeta-function $\zT$ \eqref{zT}. One proceeds as in Proposition \ref{prop:truncated} and substitutes $\trunc{\heat}_P$ for $\heat_P$ and $\trunc{\ZZ}$ for $\ZZ$ in all of the assertions. 
\end{remark}

\section{Examples}\label{sec:examples}

\subsection{Operators of polynomial spectrum}\label{sec:polynomial}

In this section we investigate heat traces associated with operators, the eigenvalues and multiplicites of which are given by polynomials. For brevity we shall call such operators -- the \emph{operators of polynomial spectrum}. They appear naturally in the context of Dirac and Laplace operators on spheres \cite{SphereDirac,Trautman} and their isospectral deformations \cite{ConnesSU2,ConnesIsospectral,AllPodles,EquatorialPodles,DiracSUq2,PalSundar}.

Moreover, the results presented in this section apply directly in the framework of Dirac operators on other homogeneous spaces \cite{BarHomogeneous,Marcolli1,TehPHD}, like the Poincar\'e sphere or lens spaces. Indeed, recall that the spectra of the relevant operators can be written as \cite{BarHomogeneous,TehPHD}
\begin{align*}
\sigma(P) = \bigcup_{k = 1}^{N} (\lambda_n^k)_{n \in \N},
\end{align*}
for some finite $N \in \Np$, and for each $k \in \Np$ the eigenvalues $\lambda_n^k$ and respective multiplicites $M_n^k$ are given by polynomials in $n$. Since the (inverse) Mellin transform is linear, one can apply the general theory to each sequence $(\lambda_n^k)_{n \in \N}$ separately.

By combining the results of \cite{MatsumotoWeng} with general theorems presented in the preceding section we are able to prove the existence of an asymptotic expansion of the heat traces associated with operators of polynomial spectra (see Theorem \ref{thm:heat_poly}). We also derive sufficient conditions for the convergence of the expansion.

\subsubsection{Asymptotic expansions}

Let us start with the following theorem summarising the behaviour of zeta-functions associated with operators in the considered class (compare \cite[Theorems A and B]{MatsumotoWeng}).

\begin{theorem}\label{thm:zPoly}
Let $P$ be an operator with eigenvalues $\lambda_{n} = A(n)$ and multiplicities $M_n = B(n)$, where $A$ and $B$ are polynomials. Assume moreover that the roots of $A$ are not in $\N$.

Then:
\begin{enumerate}
	\item \label{zPoly_def} $\zP$ is well-defined, with the abcissa of convergence $L = (1+\fb)/\fa$;
	\item \label{zPoly_ext} $\zP$ admits a meromorphic extension to the whole complex plane;
	\item \label{zPoly_poles} $\PP_{\zP}(\C) \subset \tfrac{1}{\fa}(1+\fb - \N) \setminus (-\N)$ and all of the poles are of first order;
	\item \label{zPoly_est} For any $\epsilon >0$ there exists $N>0$ such that the truncated zeta-function (see \eqref{zT}), $\zT$, obeys the following growth rate along the vertical lines
		\begin{align*}
			\abs{\zT(x+iy)} = \Oinf\left( e^{ \epsilon \abs{y}} \right),
		\end{align*}
		for any $x \in \R$.
\end{enumerate}
\end{theorem}

\begin{proof}
Let us start with point \ref{zPoly_def}. The zeta-function associated with the operator $P$ for $\Re(s) > L$ can be written as
\begin{align}
\zP(s) = \sum_{n=0}^{\infty} \frac{B(n)}{A(n)^{s}},
\end{align}
since the roots of $A$ are not in $\N$. The corresponding spectral growth function reads
\begin{align*}
N(\lambda_n) = \sum_{k=0}^n B(k)=:\wt{B}(n),
\end{align*}
and it is a classical result (see e.g. the Faulhaber's formula) that $\wt{B}$ is a polynomial and $\deg \wt{B}=\deg B +1$. Hence, by Theorem \ref{thm:abscissa} the abscissa of convergence of $\zP$ equals (cf. \eqref{limsup_lambda})
\begin{equation}
L=\limsup_{n\to\infty}\tfrac{\log N(\lambda_n)}{\log\lambda_n}
=\limsup_{n\to\infty}\tfrac{\log \wt{B}(n)}{\log A(n)}
=\tfrac{\deg B +1}{\deg A}.
\end{equation}

Point \ref{zPoly_ext} is the content of Theorem B in \cite{MatsumotoWeng} (see also Remark 1 therein).

To prove points \ref{zPoly_poles} and \ref{zPoly_est} we introduce for $r\in\Np$ the following multi-variable series (see \cite[Section~1, Formula (2)]{MatsumotoWeng})
\begin{align}\label{zeta_r}
\zeta_r\left( s_1, \ldots, s_r; \alpha_1,\ldots,\alpha_r \right) \vc \sum_{n=0}^{\infty} (n+\alpha_1)^{-s_1} (n+\alpha_2)^{-s_2} \cdots (n+\alpha_r)^{-s_r},
\end{align}
for $s_1,\ldots,s_r \in \C$ with $\Re(s_1 + \ldots + s_r) > 1$ and $\alpha_1, \ldots, \alpha_r \in \C \setminus (- \N)$. The branch of logarithm in $(n+\alpha_j)^{-s_j} = \exp \left( -s_j \, \log(n+\alpha_j) \right)$ is chosen to be $-\pi < \arg(n+\alpha_j) \leq \pi$.

Note that if we write
\begin{align*}
A(n) = a \prod_{i=1}^{\fa} (n+\alpha_i), && B(n) = \sum_{j=0}^{\fb} \wt{b}_j (n+\alpha_1)^j,
\end{align*}
then (compare \cite[Section~1, Formula (3)]{MatsumotoWeng})
\begin{align}\label{zeta_conversion}
\zP(s) = a^{-s} \sum_{j=0}^{\fb} \wt{b}_j \zeta_{\fa}\left(s-j,s, \ldots, s; \alpha_1,\ldots,\alpha_r \right).
\end{align}

In \cite{MatsumotoWeng} the analytic properties of $\zeta_r$ functions are studied, what allows to draw conclusions about the analytic properties of $\zP$ on the strength of formula \eqref{zeta_conversion}. formulae (10) and (12) in \cite[Section 1]{MatsumotoWeng} provide an explicit meromorphic continuation of $\zeta_r$ to $\C^r$. The proof is based on the induction on $r$ with $r=1$ --- the Riemann zeta-function --- as a starting point.

The claim $\PP_{\zP}(\C) \subset \tfrac{1}{\fa}(1+\fb - \N)$ then follows from \cite[Theorem A]{MatsumotoWeng}. On the other hand, \cite[Theorem C]{MatsumotoWeng} implies that $-\N \nsubseteq \PP_{\zP}(\C)$. The fact that the poles $\zP$ are at most of first order is a consequence of \cite[Lemma 5]{MatsumotoWeng} and the formula \cite[Section~1, Formula (12)]{MatsumotoWeng} (see also \cite[Section 1.2 and p. 242]{MatsumotoWeng}).

Let us now pass on to the last point of the claim -- the estimate of $\zT$ on vertical lines. We first quote the result \cite[Proposition 1 (iii)]{MatsumotoWeng} translated to our notation:
\begin{align*}
\abs{\zeta_r\left( x_1 + i y_1, \ldots, x_r + i y_r; \alpha_1,\ldots,\alpha_{r-1},0 \right)} = 
\Oinf \left( C(y_1,\ldots,y_r) \. e^{\rho_1 |y_1| + \ldots + \rho_{r-1} |y_{r-1}|} \right),
\end{align*}
with $\rho_i = \abs{\arg \, \alpha_i}$ and $C$ -- a polynomial. Now, we repeat the proof of Proposition 1 (iii) \cite[p. 240]{MatsumotoWeng} with the formula (12) instead of (10) therein. The reasoning goes along the same lines, \cite[Lemma 1]{MatsumotoWeng} still applies with $y = \Im(z), \, A = - \tfrac{\pi}{2}, \, \alpha = y_r, \, B = \rho_r - \tfrac{\pi}{2}, \, \beta = 0$ and we conclude that
\begin{align}\label{zeta_poly_bound}
\abs{\zeta_r\left( x_1 + i y_1, \ldots, x_r + i y_r; \alpha_1,\ldots,\alpha_{r-1},\alpha_r \right)} = 
\Oinf \left( C(y_1,\ldots,y_r) \. e^{\rho_1 |y_1| + \ldots + \rho_{r-1} |y_{r-1}| + \rho_{r} |y_{r}|} \right).
\end{align}
Hence, by \eqref{zeta_conversion} we obtain
\begin{align*}
\abs{\zP(x + i y)} & = \Oinf \left( C(y)^{\fa} \. e^{\abs{y}\sum_{i=1}^{\fa} \abs{\arg \, \alpha_i}} \right) \\
 & = \Oinf \left( e^{\abs{y} \left(\sum_{i=1}^{\fa} \abs{\arg \, \alpha_i} + \delta \right)} \right),
\end{align*}
for any $\delta > 0$.

Here comes the advantage of using the truncated zeta-function $\zT$, since by starting the zeta-series at $n=N$ instead of $n=0$ we effectively shift $\alpha_i \to \alpha_i + N$ (compare \cite[Remark 1]{MatsumotoWeng}). This means that by taking $N$ large enough one can make $\sum_{i=1}^{\fa} \abs{\arg \, \alpha_i}$ arbitrarily small and assertion \ref{zPoly_est} is proven.
\end{proof}

Typically (i.e. for generic polynomials $A$ and $B$) the zeta-functions $\zP$ will have an infinite number of poles on the negative part of the real axis. This is a consequence of the \cite[formulae (10) and (12)]{MatsumotoWeng}. Let us illustrate this property with the following simple example.

\begin{example}
Let $P$ be an operator with eigenvalues $\lambda_n = A(n) = n(n+\alpha)$ for $n \in \Np$ with some $\alpha \in \R^+$ and no degeneracies (i.e. $B(n) = 1$), then $\PP_{\zP}(\C) = \tfrac{1}{2} - \N$.
\end{example}
Indeed, from the formulae \eqref{zeta_conversion} and \cite[(10)]{MatsumotoWeng} we have for $\Re(s) > -M/2$ with any $M \in \N$:
\begin{align*}
\zP(s) = \sum_{j=0}^{M} \binom{-s}{j} \zeta(2s+j) \alpha^j + h_M(s),
\end{align*}
where $h_M$ is a remainder term holomorphic for $\Re(s) > -M/2$. So for any $n \in \N$ we have
\begin{align*}
\Rez{s=\tfrac{1}{2} - n} \zP(s) & = \Rez{s=\tfrac{1}{2} - n} \left( \sum_{j=0}^{2n} \binom{-s}{j} \zeta(2s+j) \alpha^j + h_{2n}(s) \right) = \binom{n-\tfrac{1}{2}}{2n} \alpha^{2n} \neq 0.
\end{align*}

There exist however special operators $P$ with fine-tuned polynomials $A$ and $B$ for which the zeta-function $\zP$ will only have a finite number of poles. It happens for instance in the following case:

\begin{Prop}
If all of the roots of the polynomial $A$ are equal then the set $\PP_{\zP}(\C)$ is finite.
\end{Prop}
\begin{proof}
Since $A(n) = a (n+\alpha)^{\fa}$ we have for $\Re(s) > (1+\fb)/\fa$
\begin{align*}
\zP(s) = \sum_{n=0}^{\infty} \sum_{j=0}^{\fb} \wt{b}_j \, a^{-s} \, (n+\alpha)^{-(\fa) s + j} = a^{-s} \sum_{j=0}^{\fb} \wt{b}_j \, \zeta_H((\fa) s - j,\alpha),
\end{align*}
where $\zeta_H$ is the Hurwitz zeta-function. Since $\zeta_H$ is meromorphic on $\C$ with a single simple pole at 1 and the sum over $j$ for $\zP$ is finite, the assertion follows.
\end{proof}

The equality of all roots of $A$ is a sufficient condition for $\PP_{\zP}(\C)$ to be finite, but not a necessary one. Consider for instance the operator $P$ with eigenvalues $\lambda_n = A(n) = n^3 + 1$ and degeneracies $B(n) = n^2$. Then, by using
\begin{align*}
(1+x)^{-s} =\sum_{k=0}^\infty \binom{-s}{k} x^k, \qquad \text{for } \abs{x} < 1,
\end{align*}
we obtain
\begin{align*}
\zP(s) = \sum_{n=0}^{\infty} \frac{n^2}{(n^3+1)^s} = \sum_{n=0}^{\infty} n^{-3s+2} (1+n^{-3})^{-s} = \sum_{j=0}^\infty \binom{-s}{j} \zeta(3s+3j-2),
\end{align*}
which gives a meromorphic extension of $\zP$ to the whole complex plane. On the other hand, the Theorem \ref{thm:zPoly} \ref{zPoly_poles} implies that $\PP_{\zP}(\C) \subset \tfrac{1}{3}(3 - \N) \setminus (-\N)$. Moreover, the poles of $\zP$ come only from the poles of $s \mapsto \zeta(3s+3j-2)$. But at $s = \tfrac{k}{3} - n$, $\zeta(k - 3n + 3j -2)$ is finite for all $n \in \N$, $j \in \N$ and $k \in \{1,2\}$ since the only pole of the Riemann zeta-function is at 1. Hence, $\zP$ has only one simple pole $\PP_{\zP}(\C) = \{1\}$.

It is interesting to compare this result with \cite[Lemma 1.10.1]{Gilkey1} and \cite[p. 2]{GilkeyGrubb}. The former tells us that zeta-functions associated with classical positive elliptic differential operators on (finite dimensional) compact manifolds have a finite number of poles only. On the other hand, the latter says that positive elliptic classical pseudodifferential operators generically do have an infinite number of simple poles, including at $-\N$.

Clearly, elliptic pseudodifferential operators on compact manifolds need not be of polynomial spectrum --- for instance, the scalar Laplacian on $\mathbb{T}^2$ \cite{TorusSA}. On the other hand, one might ask whether any operator of polynomial spectrum can be realised as a classical elliptic pseudodifferential operator on some compact manifold. We consider it as an interesting open problem, the solution to which may shed more light on the geometrical meaning of the operators in this class.

Finally, let us turn to the heat traces of operators with eigenvalues and multiplicites given by polynomials. The following Theorem establishes the existence of an asymptotic expansion of heat trace for any operator in this class.

\begin{theorem}\label{thm:heat_poly}
Let $P$ be an operator with eigenvalues $\lambda_{n} = A(n)$ and multiplicities $M_n = B(n)$, where $A$ and $B$ are polynomials of degree $\fa$ and $\fb$ respectively. Assume moreover that the roots of $A$ are not in $\N$.

Then, there exists $N \in \N$ such that
\begin{align}\label{heat_poly}
\heat_P(t) -  \sum_{n=0}^{N-1} B(n) e^{A(n) t} \tzero \sum_{k=0}^{\infty} \, \Rez{s = \ell - k/(\fa)} \left( \Gamma(s) \zT(s) \right) t^{\ell - k/(\fa)},
\end{align}
where $\ell = (1 + \fb)/(\fa)$.
\end{theorem}
\begin{proof}
The claim follows from a direct application of Theorem \ref{thm:hk_asymptotic_general} together with Remark \ref{remark_trunc}. The assumptions are met by Theorem \ref{thm:zPoly}.
\end{proof}

\begin{remark}
Let us remark that one could actually extend Theorem \ref{thm:zPoly} to operators, the eigenvalues and multiplicites of which can be written as
\begin{align*}
\lambda_n = a n^{\gamma_0} \prod_{i =1}^{\fa} (n^{\gamma_i}+\alpha_i), && M_n = \sum_{j=0}^{\fb} b_j n^{\beta_j},
\end{align*}
with $\gamma_i > 0$ for  $i =1,\ldots,\fa$, $\sum_{i =0}^{\fa} \gamma_i > 0$ and $b_j > 0$ for at least one $j \in \{0, \ldots, \fb\}$. Then, instead of \eqref{zeta_r}, one would need to seek for a meromorphic continuation to $\C^r$ of the functions
\begin{align*}
\wt{\zeta}_r\left( s_1, \ldots, s_r; \gamma_0, \ldots, \gamma_r; \alpha_1,\ldots,\alpha_r \right) \vc \sum_{n=0}^{\infty} n^{-s_0 \gamma_0} (n^{\gamma_1}+\alpha_1)^{-s_1} \cdots (n^{\gamma_r}+\alpha_r)^{-s_r}.
\end{align*}
The latter could again be accomplished with the help of the Mellin-Barnes formula as in \cite{MatsumotoWeng}. It is plausible that $\zP$ functions obtained in this way will only have first order isolated poles and some bound similar to \eqref{zeta_poly_bound} can be established. On the other hand, the regularity of $\zP$ at $-\N$ is not to be expected in general.

We have chosen to formulate Theorem \ref{thm:zPoly} for operators of polynomial spectrum, as considering more general ones described above would add to the complexity without being strongly motivated. In fact, we were not able to find any reasonable geometric example where the operator falls into this larger class, but not the one of operators of polynomial spectrum.
\end{remark}

Generically, the expansion \eqref{heat_poly} will only be an asymptotic one. There exist, however, operators for which the formula \eqref{heat_poly} will be exact or almost exact on $t \in (0,T)$ for some $0<T\leq +\infty$.

\subsubsection{Exact expansions}

In this subsection we consider a class of operators, which have an exact expansion of the associated heat traces. In particular, it will serve as an illustration for Theorem \ref{cor:hk_exact}.

\begin{Prop}\label{prop:poly_exact}
Let $P$ be an operator of polynomial spectrum with eigenvalues $\lambda_{n} = A(n)$ and multiplicities $M_n = B(n)$.
If $\fa = 1$, i.e. if $A(n) = a (n+\alpha)$ for some $a >0$, $\alpha \geq 0$, then the asymptotic expansion \eqref{heat_poly} is absolutely exact on $(0,2\pi/a)$.
\end{Prop}

\begin{proof}
To start, let us consider an operator $P_0$ with $a = 1$, $\alpha = 1$ and $B(n) = 1$. We can calculate directly:
\begin{align*}
\zP(s)&=\sum_{n=1}^\infty n^{-s}=\zeta(s),&&\Re(s)>1.
\end{align*}

The Riemann zeta-function $\zeta$ extends meromorphically to $\C$ with a single simple pole at 1 whereas $\Gamma$ has simple poles at non-positive integer numbers. The values of the residues read
\begin{equation*}
\Rez{s=1}\zeta(s)=1,\quad\quad\Rez{s=-n}\Gamma(s)=\tfrac{(-1)^n}{n!}\quad\forall n\in\N.
\end{equation*}
Moreover, the values of $\zeta$ at nonpositive integers are given by
\begin{align*}
\zeta(-n)=-\tfrac{B_{n+1}}{n+1},\quad\forall n\in\N,
\end{align*}
where $B_n$ denote the Bernoulli numbers with the convention $B_1=+1/2$. As a consequence, $\zeta$ vanishes at negative even integers, so the function $\ZZ(s)=\Gamma(s)\zeta(s)$ is in fact regular at these points, and the set of its poles is $\PP_\ZZ(\C)=\{1, 0, -1, -3, -5, \dotsc\}$. Theorem \ref{thm:heat_poly} therefore yields
\begin{align}
\label{Riemann_zeta_expansion}
\heat_{P_0}(t) \tzero t^{-1}-\tfrac12+\sum_{n=0}^\infty\tfrac{B_{2n+2}}{(2n+2)!}t^{2n+1}.
\end{align}

To check if this expansion is exact, by Theorem \ref{cor:hk_exact}, we need to find an explicit bound of the form \eqref{exact} for a suitable sequence $(R_n)$ tending monotonically to $+\infty$. We can, for instance, choose
\begin{align}\label{Rn_poly_exact}
R_0 = -\tfrac{3}{2}, && R_1 = -\oh, && R_2 = \oh, && R_n = 2(n-2), \text{ for } n \geq 3,
\end{align}
as $\ZZ$ is regular at negative even integers. Now, recall the Riemann functional equation \cite[Formula (23.2.6)]{Abram}:
\begin{equation*}
\zeta(s)=2^s\pi^{s-1}\sin\left(\tfrac{\pi s}2\right)\Gamma(1-s)\zeta(1-s).
\end{equation*}
By changing $s$ to $(1-s)$ we get
\begin{equation*}
\ZZ(s)=\Gamma(s)\zeta(s)=\frac{2^{s-1}\pi^s\zeta(1-s)}{\sin\big(\pi(1-s)/2\big)}.
\end{equation*}
For $s=-R_n+iy$ with $n \geq 3$ the denominator of the above expression equals
\begin{equation*}
\sin\big(\pi(\tfrac12+n-i\tfrac y2)\big)=
\cos\big(\pi(n-i\tfrac y2)\big)=
(-1)^n\cos\big(-i\tfrac{\pi y}2)\big)=
(-1)^n\cosh\left(\tfrac{\pi y}2\right).
\end{equation*}
Thus, we have
\begin{equation*}
\abs{\ZZ(-R_n+iy)}=\frac{2^{-2n-1}\pi^{-2n}\abs{\zeta(2n+1-iy)}}{\cosh\left(\tfrac{\pi y}2\right)}.
\end{equation*}
For any $y\in\R$, $2\cosh(y)>e^{\abs{y}}$ and if $x>1$ then
\begin{equation*}
\abs{\zeta(x+iy)}\leq\sum_{k=1}^\infty \abs{k^{-x-iy}}=\sum_{k=1}^\infty k^{-x}=\zeta(x).
\end{equation*}
Hence,
\begin{equation*}
\abs{\ZZ(-2n+iy)}\leq(2\pi)^{-2n}\zeta(2n+1)e^{-\tfrac{\pi\abs{y}}2},
\end{equation*}
what means that $\ZZ$ satisfies the assumptions of Theorem \ref{thm:hk_asymptotic_general} with $C_n\vc(2\pi)^{-2n}\zeta(2n+1)$ and a constant $\epsilon_n\vc\pi/2$. Now,
\begin{equation*}
\limsup_{n\to\infty}\sqrt[R_n]{C_n/\epsilon_n}=\tfrac1{2\pi}<\infty,
\end{equation*}
as $\zeta(x)\to1$ for $x\to\infty$. Thus, on the strength of Theorem \ref{cor:hk_exact}, we see that the expansion \eqref{Riemann_zeta_expansion} is exact on $(0,2\pi)$. Moreover, with the choice \eqref{Rn_poly_exact} of the sequence $(R_n)$, Proposition \ref{prop:abs} applies and the expansion \eqref{Riemann_zeta_expansion} is in fact absolutely exact on $(0,2\pi)$.

Let us now turn to the general case of an operator $P$ with $A(n) = a (n+\alpha)$ and $B(n) = \sum_{j=0}^{\fb} b_j n^j$. For any $t>0$ we have
\begin{align*}
\heat_P(t) = \sum_{n=0}^{\infty} \sum_{j=0}^{\fb} b_j n^j e^{-t a (n+\alpha)} = e^{-t a \alpha} \, \sum_{j=0}^{\fb} b_j \sum_{n=0}^{\infty} n^j e^{-t a n} = e^{-t a \alpha} \, \sum_{j=0}^{\fb} b_j \frac{d^j}{d t^j} \heat_{P_0}(a t).
\end{align*}
Since $\R^+ \ni t \mapsto \heat_{P_0}(t)$ is an analytic function with an absolutely exact expansion for $t \in (0,2 \pi)$, we conclude that the expansion of $\heat_P$ is absolutely exact for $(0,2\pi/a)$.
\end{proof}

\begin{remark}
Since the eigenvalues of the operator at hand grow linearly, the heat trace 
\begin{align*}
\heat_P(t) = \sum_{n=0}^{\infty} B(n) e^{-t A(n)}
\end{align*}
associated with $P$ can be summed explicitly -- for instance, $\heat_{P_0}(t) = \tfrac{1}{2} \left( \coth \tfrac{t}{2} - 1 \right)$ for any $t>0$. Hence, $\heat_P$, as a linear combination of derivatives of $\heat_{P_0}$, is actually a complex analytic function of $t$ around 0. In particular, it admits a Laurent expansion around $t=0$ with the radius of convergence precisely equal to $2\pi/a$ (recall that $\coth\tfrac t2$ is singular at $t=2\pi i$).
\end{remark}

We conclude that in this particular subclass the asymptotic expansion of $\heat_P(t)$ as $t \downarrow 0$ obtained from the Mellin transform is equal to the Laurent expansion of $\heat_P$ and the value of abscissa of convergence of the former precisely coincides with the value of radius of convergence of the latter. In particular, it shows that with the estimates adopted in the statement of Theorem \ref{cor:hk_exact} one can obtain the \emph{maximal} region of convergence of the expansion of a heat trace.

Proposition \ref{prop:poly_exact} has a direct geometrical application:
\begin{Cor}\label{cor:S1_absD}
Let $\DD$ be the Dirac operator acting on a spinor bundle of $S^d$ --- the $d$-dimensional sphere with round metric. Then, the asymptotic expansion of the heat trace associated with $\abs{\DD}$ converges to $\heat_{\abs{\DD}}$ on $(0,2\pi)$.
\end{Cor}
\begin{proof}
Recall first that on $S^1$ there are two possible spin structures, whereas for $S^d$ with $d \geq 2$ there is only one available \cite{BarSpin}. In the case of the trivial spin structure on $S^1$ one has 
\begin{align*}
\lambda_{\pm n}(\DD) &= \pm n, & M_{\pm n} (\DD) &= 1,&&\text{for }n\in\N,\\
\intertext{hence}
\lambda_n(\abs{\DD}) &= n, & M_n(\abs{\DD}) &= 2 - \delta_{n,0},&&\text{for }n\in\N.
\end{align*}
So, for $S^1$ with the trivial spin structure we have $\heat_{\abs{\DD}} = 2 \heat_{P_0} + 1$, where $P_0$ is as in the proof of Proposition \ref{prop:poly_exact}. It is a special case of truncated zeta-function use with $N=1$ that has to be used because of the nontrivial kernel of $\DD$. The truncated heat trace \eqref{truncated_htr} is equal to $2 \heat_{P_0}$.

In the case of the non-trivial spin structure, the spectrum of the Dirac operator $\DD$ fits into the general pattern for $S^d$  \cite{BarSpin,SphereDirac,Trautman}:
\begin{align}\label{eigenvalues_Sd}
\lambda_{\pm n}(\DD)&=\pm(n+\tfrac d2), & M_{\pm n}(\DD)&=2^{\floor{\tfrac d2}} \tbinom{n+d-1}{d-1},&&\text{for }n\in\N.\\
\intertext{Hence, the operator $\abs{\DD}$ with}
\lambda_n(\abs{\DD})&= n+\tfrac d2, & M_n(\abs{\DD})&=2^{\floor{\tfrac d2}+1} \tbinom{n+d-1}{d-1},&&\text{for }n\in\N\notag,
\end{align}
meets the conditions of Proposition \ref{prop:poly_exact} with $a = 1$, $\alpha = \tfrac d2$ and $B(n) = M_n$ and the assertion follows.
\end{proof}


\subsubsection{Almost exact expansions}

In this subsection we turn to the case of almost exact (but not exact) expansions of heat traces associated with operators of polynomial spectrum.

\begin{Prop}\label{prop:poly_almost}
Let $P$ be an operator with eigenvalues $\lambda_{n} = A(n)$ and multiplicities $M_n = B(n)$ for $n \in \N$. If $A(n) = a (n+\alpha)^{2k} + \beta$ for some $a \in \R^+$, $\alpha \in \tfrac {1}{2} \N$, $\beta \in \R$, $k \in \Np$ and $B(n) = \sum_{j=0}^{(\fb)/2} \wt{b}_j (n+\alpha)^{2j}$ with $\fb$ even, then the asymptotic expansion \eqref{heat_poly} is almost exact for all $t>0$, but not exact.
\end{Prop}
\begin{proof}
At first let us note that with $P = P' + \beta$ we have
\begin{align*}
\heat_P(t) = e^{-t \beta} \heat_{P'}(t).
\end{align*}
Thus, it is sufficient to consider the case $\beta =0$.

Moreover, we can assume that $\alpha = \tfrac{1}{2}$ or $\alpha = 0$, and $\beta=0$. Indeed, let us denote by $\wt{\alpha} = \alpha - \floor{\alpha}$ the fractional part of $\alpha$, then we have
\begin{align*}
\heat_{P'}(t) & = \sum_{n=0}^{\infty} B(n) e^{-t A(n)} = \sum_{j=0}^{(\fb)/2} \wt{b}_j \, \sum_{n=0}^{\infty} (n+\alpha)^{2j} e^{-t a (n+\alpha)^{2k}} = \\
& = \sum_{j=0}^{(\fb)/2} \wt{b}_j \, \sum_{n=0}^{\infty} (n+ \floor{\alpha} + \wt{\alpha})^{2j} e^{-t a (n+ \floor{\alpha} + \wt{\alpha})^{2k}} =\\
& = \sum_{j=0}^{(\fb)/2} \wt{b}_j \, \sum_{n=\floor{\alpha}}^{\infty} (n+ \wt{\alpha})^{2j} e^{-t a (n+ \wt{\alpha})^{2k}} =\\
& = \sum_{j=0}^{(\fb)/2} \wt{b}_j \, \sum_{n=0}^{\infty} (n+ \wt{\alpha})^{2j} e^{-t a (n+ \wt{\alpha})^{2k}} 
- \sum_{j=0}^{(\fb)/2} \wt{b}_j \, \sum_{n=0}^{\floor{\alpha} - 1} (n+ \wt{\alpha})^{2j} e^{-t a (n+ \wt{\alpha})^{2k}}
\end{align*}
and the second term is regular at $t = 0$.

Let us denote by $P''$ the operator with $\lambda_n(P'') = a (n+\wt{\alpha})^{2k}$ and $M_n(P'') = M_n(P)$. For $\Re(s) > (1+\fb)/(\fa)$ we have 
\begin{align*}
\zeta_{P''}(s) = \sum_{n=0}^{\infty} \sum_{j=0}^{(\fb)/2} \wt{b}_j a^{-s} (n+\wt{\alpha})^{2j - 2ks} = a^{-s} \sum_{j=0}^{(\fb)/2} \wt{b}_j \zeta_H(2ks - 2j,\wt{\alpha}),
\end{align*}
where $\zeta_H$ is the Hurwitz zeta-function. This formula provides a meromorphic extension of $\zeta_{P''}$ to the whole complex plane. Since the Hurwitz zeta-function has a single simple pole at 1 for any $\wt{\alpha} \in \C\setminus - \N$, we have
$\PP_{\zeta_{P''}}(\C) = \tfrac{1}{2k} \{1,3,\ldots,\fb+1\}$. What is more, for $s = -n$ with $n \in \N$ we have
\begin{align*}
\zeta_{P''}(-n) = a^n \sum_{j=0}^{(\fb)/2} \wt{b}_j \zeta_H(-2kn - 2j,\wt{\alpha}) = - a^n \sum_{j=0}^{(\fb)/2} \wt{b}_j \frac{B_{2kn+2j+1}(\wt{\alpha})}{2kn+2j+1},
\end{align*}
where $B_n$ denotes the Bernoulli polynomial of degree $n$. But if $\wt{\alpha} = 0$ or $\wt{\alpha} = \tfrac{1}{2}$ then $B_{2n+1}(\wt{\alpha}) = 0$ \cite[Section 23]{Abram}, hence for any $n \in \Np$ $\zeta_{P''}(-n) = 0$. So if $\ZZ(s) = \Gamma(s) \zeta_{P''}(s)$ as usually, we have $\PP_\ZZ(\C) = \tfrac{1}{2k} \{0,1,3,\ldots,\fb+1\}$. Thus, on the strength of Proposition \ref{prop:finite} and formula \eqref{heat_poly} we finally conclude that for all $t>0$
\begin{multline}\label{gen_Jacobi}
\heat_P(t)  =  e^{-t \beta} \, \sum_{j=0}^{(\fb)/2} \wt{b}_j \, \Gamma \left( \tfrac{1+2j}{2k} \right) (at)^{- \tfrac{1+2j}{2k}} - e^{-t \beta} \wt{b}_0 \, \left( \wt{\alpha} - \tfrac{1}{2} \right) +\\
- e^{-t \beta} \, \sum_{j=0}^{(\fb)/2} \wt{b}_j \, \sum_{n=0}^{\floor{\alpha} - 1} (n+ \wt{\alpha})^{2j} e^{-t a (n+ \wt{\alpha})^{2k}} + F_{\infty}(t),
\end{multline}
where $F_{\infty}(t) \neq 0$ has a null Taylor expansion at $t=0$.
\end{proof}

In some cases, one could be able to compute the remainder function $F_{\infty}$ explicitly.
\begin{example}\label{ex:theta}
Let $P$ be an operator with $\lambda_n(P) = n^2$ and $M_n = 1$ for $n \in \N$, then
\begin{align*}
\heat_{P}(t) = \sum_{n=0}^{\infty} e^{-t n^2} = \frac12(\theta_3(0;e^{-t})-1),
\end{align*}
with $\theta_3$ being the Jacobi theta function defined as 
\begin{equation*}
\theta_3(z;q)\vc\sum_{n=-\infty}^\infty q^{n^2}e^{2niz}.
\end{equation*}
For such an operator one can compute explicitly (see \cite{Elizalde89} and also \cite{PhD}):
\begin{align*}
F_{\infty}(t) = \sqrt{\tfrac\pi t}\tfrac12\left(-1+\theta_3\left(0;e^{-\pi^2/t} \right)\right).
\end{align*}
Hence, formula \eqref{gen_Jacobi} turns out to be the famous Jacobi identity \cite[(3.13)]{Elizalde89}. This example has been analysed in details in \cite{Elizalde89} and has direct applications to quantum physics.
\end{example}

Therefore, one can regard the formula \eqref{gen_Jacobi} as a generalisation of the Jacobi identity for special functions defined by $\heat_P(t)$ with $P$ as in Proposition \ref{prop:poly_almost}.

Proposition \ref{prop:poly_almost} again has a direct geometrical application.
\begin{Cor}\label{cor:S1_D2}
Let $\DD$ be the Dirac operator acting on a spinor bundle of $S^d$ -- the $d$-dimensional sphere with round metric. If $d$ is odd, then the asymptotic expansion of the heat trace associated with $\DD^{2k}$ with $k\in\Np$ is almost exact for all $t > 0$, but not exact.
\end{Cor}
\begin{proof}
To apply Proposition \ref{prop:poly_almost} we need to show that $M_n(\DD^{2k})$ can be written as a polynomial in $\lambda_n(\DD^2)$. For trivial spin structure on $S^1$ it is obvious, however we need to resort again to the truncated zeta function, as in Corollary \ref{cor:S1_absD}. 

In all other cases, by formula \eqref{eigenvalues_Sd} we have $\lambda_n(\DD^{2k}) = (n + \tfrac{d}{2})^{2k}$. Note that
\begin{multline}\label{Mn_spheres}
M_n(\DD^{2k}) = M_n(|\DD|)\\
= 2^{\floor{\tfrac d2}+1} (\lambda_n(\abs{\DD})+\tfrac d2 -1)(\lambda_n(\abs{\DD})+\tfrac d2 -2)\dotsm(\lambda_n(\abs{\DD})-\tfrac d2 +2)(\lambda_n(\abs{\DD})-\tfrac d2 +1).
\end{multline}
So if $d$ is odd, then
\begin{align*}
M_n(\DD^{2k}) & = 2^{\floor{\tfrac d2}+1} \bigg(\lambda_n(\abs{\DD})^2-(\tfrac d2 -1)^2\bigg)\bigg(\lambda_n(\abs{\DD})^2-(\tfrac d2 -2)^2\bigg)\dotsm\bigg(\lambda_n(\abs{\DD})^2-(\tfrac 12)^2\bigg).
\end{align*}
Hence, $M_n(\DD^{2k})$ can indeed be written as a polynomial in $\lambda_n(\abs{\DD})^2 = \lambda_n(\DD^2)$.
\end{proof}

The comparison of Corollaries \ref{cor:S1_absD} and \ref{cor:S1_D2} shows a sharp contrast between the heat traces associated with the operators $\abs{\DD}$ and $\DD^2$ for \textit{odd-dimensional} spheres. For the even dimensional ones the discrepancy is even more dramatic, as we shall now see.

\subsubsection{Divergent expansions}\label{divergent}

For a generic operator of polynomial spectrum with eigenvalues growing at least quadratically one expects the associated heat trace expansion to be only an asymptotic one. In particular, we have the following result:

\begin{Prop}
Let $\DD$ be the Dirac operator acting on a spinor bundle of $S^{d}$ -- the $d$-dimensional sphere with round metric. If $d$ is even, then the asymptotic expansion of the heat trace associated with $\DD^2$ is only asymptotic.
\end{Prop}

\begin{proof}
If $d$ is even, formula \eqref{Mn_spheres} yields
\begin{align*}
M_n(\DD^{2}) & = 2^{\floor{\tfrac d2}+1} \bigg(\lambda_n(\abs{\DD})^2-(\tfrac d2 -1)^2\bigg)\bigg(\lambda_n(\abs{\DD})^2 -(\tfrac d2 -2)^2\bigg)\dotsm\bigg(\lambda_n(\abs{\DD})^2-1^2\bigg)\lambda_n(\abs{\DD}) \\
& \cv \sum_{m=0}^{d-1} c_m \lambda_n(\abs{\DD})^m,
\end{align*}
with $c_m = 0$ for $m$ even. Moreover, 
\begin{align}
\label{sign_c}
\sign c_{d-1-2q}=(-1)^q, \quad\forall q=0,\dotsc,\tfrac d2-1.
\end{align}

The zeta-function associated with the operator $P = \DD^{2}$ on $S^d$ reads
\begin{align*}
\zP(s) = \sum_{n=0}^{\infty} M_n(\DD^{k}) \lambda_n(|\DD|)^{-2s} = \sum_{m=0}^{d-1} c_m \sum_{n=0}^{\infty} (n+\tfrac{d}{2})^{-2s+m} = \sum_{m=0}^{d-1} c_m \zeta_H(2s - m, \tfrac{d}{2}),
\end{align*}
where $\zeta_H$ is the Hurwitz zeta-function. It turns out, that the latter can actually be replaced by the Riemann zeta-function. Indeed, let us note that
\begin{align*}
\zeta_H(s,\tfrac{d}{2}) = \sum_{n=0}^{\infty} (n+\tfrac{d}{2})^{-s} = \sum_{n=1}^{\infty} n^{-s} - \sum_{n=1}^{d/2-1} n^{-s} = \zeta(s) - F_d(s),
\end{align*}
with $F_d(s) \vc \sum_{n=1}^{d/2-1} n^{-s}$ for $d \geq 4$ and $F_2(s) = 0$. Hence,
\begin{align*}
\zP(s) & = \sum_{m=0}^{d-1} c_m \zeta(2s - m) - \sum_{m=0}^{d-1} c_m F_d(2s - m).
\end{align*}
Let us investigate the second term on the RHS of the above equality. For any $s \in \C$,
\begin{align*}
\sum_{m=0}^{d-1} c_m F_d(2s - m) = \sum_{n=1}^{d/2-1} n^{-2s} \sum_{m=0}^{d-1} c_m n^{m} = \sum_{n=1}^{d/2-1} n^{-2s} B(n-\tfrac{d}{2}),
\end{align*}
where $B$ is the polynomial defining the multiplicites of $\DD^{2}$, i.e. $M_n(\DD^{2}) = B(n)$ for all $n \in \N$. But, since $B(n) = 2^{\floor{\tfrac d2}+1} \tbinom{n+d-1}{d-1}$ by \eqref{eigenvalues_Sd} and $d$ is even, we have $B(n-\tfrac{d}{2}) = 0$ for every $n = 1, \ldots, \tfrac{d}{2} - 1$. Thus,
\begin{align*}
\zP(s) = \sum_{m=0}^{d-1} c_m \zeta(2s - m).
\end{align*}

Now, Theorem \ref{thm:heat_poly} yields the following asymptotic expansion:
\begin{align}\label{heat_S2}
\heat_{P}(t) & \tzero \sum_{k=0}^{\infty} \Rez{s= (d-k)/2} \Gamma(s) \zP(s) t^{(k-d)/2} \\
&= \sum_{k=0}^{d-1}  \Gamma(\tfrac{d-k}{2}) c_{d-k-1} t^{(k-d)/2} + \sum_{p=0}^{\infty} \tfrac{(-1)^p}{p!} \sum_{m=0}^{d-1} c_m \zeta(-2p - m) t^p. \notag
\end{align}

Let us investigate the convergence of the series
\begin{align}\label{series_p}
\sum_{p=0}^{\infty} d_p t^p, \quad \text{ with } \quad d_p = \tfrac{(-1)^p}{p!} \sum_{m=0}^{d-1} c_m \zeta(-2p - m).
\end{align}
Recall that \cite[(23.2.15)]{Abram}
\begin{align*}
\zeta(-2p-m) = -\frac{B_{2p+m+1}}{2p+m+1},
\end{align*}
where $B_n$ are the Bernoulli numbers.

Since $c_m = 0$ for even $m$, $2p+m+1$ is always even. Moreover, we have
\begin{align*}
\sign B_{2n}=(-1)^{n+1}
\end{align*}
and by \eqref{sign_c}
\begin{equation*}
\sign c_m=(-1)^{\tfrac{d-1-m}2},
\end{equation*}
thus
\begin{equation*}
\sign c_m \,B_{2p+m+1}=(-1)^{\tfrac{d-1-m}2}(-1)^{\tfrac{2p+m+1}2+1}=(-1)^{\tfrac d2+p+1}.
\end{equation*}
Therefore,
\begin{align*}
\sign d_p = (-1)^{\tfrac d2}
\end{align*}
so all of the terms in series \eqref{series_p} over $p$ has the same signs. Moreover, since $\sign c_m \,B_{2p+m+1}$ does not depend on $m$ it is sufficient to study the behaviour of $c_m \zeta(-2p - m) / (p!)$ as $p$ grows, for a fixed $m$.

We have \cite[(23.1.15)]{Abram}
\begin{align*}
(-1)^{n+1} B_{2n} > \frac{2 (2n)!}{(2 \pi)^n},
\end{align*}
so that
\begin{align*}
c_m \frac{\zeta(-2p - m)}{p!} = c_m \frac{B_{2p+m+1}}{p! (2p+m+1)} > \frac{2\abs{c_m}}{(2\pi)^{(m+1)/2}} \, \frac{(2p+m)!}{(2 \pi)^p \, p!} \xrightarrow[p \to \infty]{} \infty.
\end{align*}

Hence, the asymptotic series in \eqref{heat_S2} diverges for any $t>0$.
\end{proof}

The lesson from the example of spheres is that if one is interested in the convergence properties of heat trace expansion one should work with $\abs{\DD}$, which is of the first order, rather than with $\DD^2$, despite the fact that the former is not a differential operator, but only a pseudodifferential one.

\subsection{Operators of exponential spectrum}\label{sec:exp}

In this section we consider the class of \emph{operators of exponential spectrum}, i.e. the ones with the spectrum $\sigma(P) = (q^{-n})_{n \in \N}$ for some $0 < q <1$. We shall also assume that the multiplicities are given by a polynomial. This type of operators appears naturally in the context of quantum groups \cite{SeniorKaad,JapanHeat} and their homogeneous spaces \cite{dab_sit,PodlesSA,NeshTuset} (see also \cite{Kaneko} and \cite[Example 12]{Flajolet}), as well as in the framework of fractal spaces \cite{Christensen2,Christensen1,Cipriani,Guido1,Guido2}.

\begin{Prop}
Let $P$ be an operator with $\lambda_n = q^{-n}$ for some $0<q<1$ and $M_n = p(n)$ for some polynomial $p$ of degree $m$. Then, the asymptotic expansion of the heat trace $\heat_P$ is absolutely exact for all $t>0$ and can be expressed as
\begin{align}\label{heat_exp}
\heat_P(t) = \tfrac{\wt{p}(1)}{(m+1)!} \, (-\log_q t)^{m+1} + \sum_{j=0}^m \left( r_j + G_j(\log_q t) \right) (\log_q t)^{j} + \sum_{n=1}^{\infty} \tfrac{\wt{p}(q^{-n})}{(1-q^{-n})^{m+1}} \, \tfrac{(-t)^n}{n!}.
\end{align}
$\log_q t = \tfrac{\log t}{\log q}$, $\wt{p}$ is a polynomial of degree $m$, $r_i$ are constants (with respect to $t$) and $G_i$ are Fourier series completely determined by $q$ and the polynomial $p$. 
\end{Prop}

\begin{proof}
We start with the analysis of the zeta-function. For $\Re(s) > 0$ we have
\begin{align*}
\zP(s) = \sum_{n=0}^{\infty} p(n) q^{ns} = \sum_{j=0}^{m} p_j  \sum_{n=0}^{\infty} n^j q^{ns} \cv \frac{\wt{p}(q^s)}{(1-q^s)^{m+1}}.
\end{align*}
The polynomial $\wt{p}$ is completely determined by $p$ via the formula
\begin{align}\label{Eulerian}
\sum_{n=0}^{\infty} n^j q^{ns} = \Li_{-j}(q^s) = \frac{1}{(1-q^{s})^{j+1}} \sum_{k=0}^{j-1} {j \bangle k}  q^{s(j-k)},
\end{align}
where $\Li$ is the polylogarithm (Jonqui\`ere's) function \cite{Jonquiere} and ${j \bangle k}$ stand for Eulerian numbers of the first kind \cite{Comtet}. From the formula \eqref{Eulerian} we also deduce that $\wt{p}(1) = p_m\cdot m! \neq 0$. Indeed,
\begin{align*}
\wt{p}(q^s) = \sum_{j=0}^{m} p_j (1-q^s)^{m-j} \sum_{k=0}^{j-1} {j \bangle k}  q^{s(j-k)}
\end{align*}
and, using the summation formula for Eulerian numbers \cite[p. 242]{Comtet},
\begin{align*}
\wt{p}(1) = p_m \sum_{k=0}^{m-1} {m \bangle k} = p_m\cdot m!.
\end{align*}

Therefore, we conclude that the function $\ZZ : s \mapsto \Gamma(s) \zP(s)$ has a meromorphic extension to the whole complex plane with:
\begin{itemize}
\item first order poles at $s \in -\Np$,
\item $m+1$ order poles at $s \in \tfrac{2 \pi i}{\log q} \Z^*$,
\item $m+2$ order pole at $s=0$.
\end{itemize}

To see that the assumptions of Theorem \ref{thm:hk_asymptotic_general} are met let us choose $R_n = n + \tfrac{1}{2}$ and $y_k^{(n)} = \tfrac{2 \pi}{\log q} (k + \tfrac{1}{2})$. Let us also denote by $\hat{p}$ the polynomial $\hat{p}(x) = \sum_{i=0}^m \vert \wt{p}_i \vert x^i$.

On the horizontal lines of the contour integration we have
\begin{align*}
\abs{\zP(x + i y_k^{(n)})} 
= \frac{\abs{\wt{p} \left(q^{x + i y_k^{(n)}} \right)}}{\abs{1-q^{x+i y_k^{(n)}}}^{m+1}} 
= \frac{\abs{\wt{p} \left(q^{x + i y_k^{(n)}} \right)}}{(1+q^x)^{m+1}}
\leq \frac{\hat{p}(q^x)}{(1+q^x)^{m+1}}
\end{align*}
and $\Gamma$ decays exponentially on verticals \cite[(2.1.21)]{Paris}, hence assumption \ref{assum4:general} of Theorem \ref{thm:hk_asymptotic_general} is fulfilled. Similarly, on the vertical lines of integration we have
\begin{align}\label{q_zeta_V}
\abs{\zP(-R_n + i y)} = \frac{\abs{\wt{p}(q^{-R_n + i y})}}{\abs{1-q^{-R_n + i y}}^{m+1}} \leq \frac{\hat{p}(q^{-R_n})}{\abs{1-q^{-R_n}}^{m+1}}.
\end{align}

To show that we have an exact expansion of $\heat_P$ valid for any $t>0$ we need to estimate the Gamma function more precisely. The Euler reflection formula \cite[(2.1.20)]{Paris} together with the estimate (see \cite[Formula (43)]{PodlesSA})
\begin{align*}
|\Gamma(x+ i y)|^{-1} \leq (2\pi)^{-1/2} \,e^{x+ |y| \,|\arg(x+iy)|}\, (x^2 + y^2)^{-x/2 + 1/4}, \quad \text{ for } x>0, y\in\R
\end{align*}
gives
\begin{align}
\abs{\Gamma(-R_n + iy)} & \leq \sqrt{\tfrac{\pi}{2}} \, \frac{1}{\cosh(\pi y)} \, e^{1+R_n+\abs{y} \abs{1+R-iy}} \left( (1+R_n)^2 + y^2 \right)^{-R_n/2-1/4} \notag \\
& \leq \sqrt{\tfrac{\pi}{2}} \, e^{1+R_n} R_n^{-R_n-1/2} \, \frac{1}{\cosh(\pi y)} e^{\abs{y} \abs{1+R-iy}} \notag \\
& \leq \sqrt{2 \pi} \, e^{1+R_n} R_n^{-R_n-1/2} \, \frac{e^{\tfrac{\pi}{2} \abs{y}}}{e^{\pi y} + e^{-\pi y}} \notag \\
& \leq \sqrt{2 \pi} \, e^{1+R_n} R_n^{-R_n-1/2} \, e^{-\tfrac{\pi}{2} \abs{y}}. \label{gamma_vert}
\end{align}
The above result together with \eqref{q_zeta_V} yields the following constants in assumption \ref{assum3:general} of Theorem \ref{thm:hk_asymptotic_general}:
\begin{align*}
\epsilon_n = \tfrac{\pi}{2}, && C_n = \sqrt{2 \pi} e^{3/2} \, \tfrac{e^n \, \hat{p}(q^{-n-1/2})}{\abs{1-q^{-n-1/2}}} \left( n + \tfrac{1}{2} \right)^{-n}.
\end{align*}

Therefore, Theorem \ref{cor:hk_exact} applies and 
\begin{align*}
\limsup_{n\to\infty}\sqrt[R_n]{C_n/\epsilon_n} = 0, &&\text{for any } 0 < q <1,
\end{align*}
yielding $T = \infty$. Moreover, since $R_n = n+ \tfrac{1}{2}$, Proposition \ref{prop:abs} applies and the expansion is absolutely exact.

The formula \eqref{heat_exp} results from a direct calculation of the residues.
\end{proof}


This result is somewhat surprising at the first sight. Clearly, $\heat_P(t)$ decays exponentially as $t$ grows to infinity, whereas the RHS of \eqref{heat_exp} seems to diverge to infinity because of the $\log t$ terms. However, it turns out that the sum over $n$ in the RHS of \eqref{heat_exp} compensates for large $t$ the divergent terms $(\log t)^j$, the constant terms and the oscillatory part to yield an exponential decay (compare \cite[Example 12]{Flajolet} and \cite[Section 4.1]{PodlesSA}).

It is interesting to note that if $P$ is an operator of exponential spectrum, then so is $P^r$ for any $r \in \R_+$ (compare also \cite[Section 4.4]{PodlesSA}). Moreover, given the exact expansion of the heat trace \eqref{heat_exp}, one immediately obtains that for $\heat_{P^r}$ simply by changing $q \leadsto q^{r}$.

\section{Outlook}

We start the concluding section of this paper with an exploration of the limitations of our general theorems. Having in mind the exactness result of the previous section one could naively expect that the faster the eigenvalues of a positive operator grow, the better the convergence properties of the associated heat trace expansion are. However, as we shall show below the exponential growth of eigenvalues establishes in fact a limit of applicability of the inverse Mellin transform.

\begin{Prop}\label{prop:lacunary}
Let $P$ be such that $M_n = \OO_{\infty}(n^b)$ for some $b \in \R^+$. If
\begin{align}\label{lacunary}
\lim_{n \to \infty} \frac{\lambda_{n+1}}{\lambda_n} = + \infty,
\end{align}
then the function $\zP$ is holomorphic for $\Re(s) > 0$, but does not admit a meromorphic continuation through $\Re(s) = 0$.
\end{Prop}
\begin{proof}
The general Dirichlet series defining the zeta-function associated with $P$ reads
\begin{align}\label{zeta_lacunary}
\zP(s) = \sum_{n = 0}^{\infty} M_n \lambda_n^{-s} = \sum_{n = 0}^{\infty} M_n e^{-s \log \lambda_n}.
\end{align}

We note that \eqref{lacunary} implies that $\lambda_n = \Oinf(e^{g(n)})$ with $g(n)/n \to +\infty$. Moreover, the assumption on the power-like growth of multiplicites assures that $\zP(s)$ is convergent for $\Re(s) \geq 0$ -- see \cite{LacunaryDirichlet2,LacunaryDirichlet} and references therein.

Under the assumption \eqref{lacunary}, the series \eqref{zeta_lacunary} is a \emph{lacunary Dirichlet series} \cite{LacunaryDirichlet}:
\begin{align*}
\zP(s) = \sum_{n = 0}^{\infty} a_n z^{\mu_n},
\end{align*}
with $z = e^{-s}$, $\mu_n = \log \lambda_n$ and $a_n = M_n$. Then, classical results \cite{Mandelbrojt} (see also \cite{Fabry}, \cite[Theorem 1]{Lacunary}, \cite{LacunaryDirichlet} and \cite{LacunaryDirichlet2}) show that the vertical line $\Re(s) = 0$ is a natural boundary of analyticity for $\zP(s)$. The latter means that the poles of $\zP$ are dense on the imaginary axis and therefore, $\zP$ cannot be extended to the left complex half-plane.
\end{proof}

We shall call the operators satisfying the assumptions of Proposition \ref{prop:lacunary} \emph{lacunary operators}. Proposition \ref{prop:lacunary} does not imply directly that $\heat_P(t)$ does not have an asymptotic expansion as $t \downarrow 0$ for lacunary operators -- in particular, both $\heat_P$ and $\zeta_P$ are well-defined. It just states that the technique of the inverse Mellin transform developed in this paper does not apply in this case. In particular, both $\heat_P$ and $\zeta_P$ are well-defined. For an a example of a lacunary operator in the realm of nonclassical pseudodifferential operators see \cite{Schrohe} \footnote{We thank Bruno Iochum for pointing out this reference to us.}.

Can anything be said about the small $t$ behaviour of heat traces associated with lacunary operators? The answer is positive and can be deduced from the following Tauberian Theorem due to Hardy and Littlewood (known also under the name of Karamata Theorem). 

\begin{theorem}[\cite{Feller} Theorem XIII.5.2 ]
\label{thm:Hardy-Littlewood}
Let $G: \R^+ \to \R$ be a function of bounded variation (see \cite[Chapter 1]{Widder}) and such that the following Riemann-Stieltjes integral
\begin{align*}
\omega(t) = \int_0^{\infty} e^{-t \lambda} dG(\lambda),
\end{align*}
converges for $t >0$. Then the following are equivalent:
\begin{alignat*}{2}
G(\lambda)&\approx\lambda^L F(\lambda),&\qquad&\text{as }\lambda\to\infty\\
\intertext{and}
\omega(t)&\approx t^{-L}\Gamma(L+1)F(t^{-1}),&&\text{as }t \downarrow 0,
\end{alignat*}
where $F$ is a slowly varying function, i.e $F(\tau x)/F(\tau)\to1$ as $\tau\to\infty$ for every $x>0$.
\end{theorem}

As a direct application of Theorem \ref{thm:Hardy-Littlewood} we obtain the following result:

\begin{Cor}
\label{cor:tauberian}
Let $P$ be a positive operator such that $\zP$ has a finite abscissa of convergence $L$. If
\begin{align*}
N(\lambda) \approx \lambda^L F(\lambda),
\end{align*}
with $F$ slowly varying, then
\begin{align*}
\heat_P(t) \approx \Gamma(L+1) t^{-L} F(t^{-1}), && \text{as } t \downarrow 0.
\end{align*}
\end{Cor}
\begin{proof}
The function $N$, being a step function, is of bounded variation \cite[Chapter 1]{Widder}. Moreover, since $\zP$ has a finite abscissa of convergence, $\heat_P$ is well-defined and thus
\begin{align*}
\heat_P(t)= \sum_{n=0}^{\infty} M_n e^{-t \lambda_n} = \int_0^\infty e^{-t\lambda}dN(\lambda)
\end{align*}
for all $t>0$. Proposition \ref{cor:zeta_abscissa} implies that for large $\lambda$ we have $N(\lambda) \approx \lambda^{L} F(\lambda)$, with $F(\lambda) = \Oinf(\lambda^{\delta})$ for every $\delta >0$. If moreover, $F$ is slowly varying, then Theorem \ref{thm:Hardy-Littlewood} applies and
the conclusion follows.
\end{proof}

Corollary \ref{cor:tauberian} applies also to lacunary operators and one can use it to detect the leading behaviour of $\heat_P$ as $t \downarrow 0$.

\begin{example}
Let us consider a lacunary operator $P$ with $\lambda_n = e^{n^2}$ and no multiplicities (i.e. $M_n\equiv1$). Its zeta-function reads
\begin{align*}
\zP(s) = \sum_{n=0}^{\infty} e^{-s n^2} = \tfrac{1}{2}(\theta_3(0;e^{-s})+1), \qquad \text{for } \Re(s) > 0,
\end{align*}
where $\theta_3$ is the Jacobi theta function we met in Example \ref{ex:theta}, but now playing the role of the zeta-function. It is a lacunary function and does not admit a meromorphic continuation to the left complex half-plane.

On the other hand, since 
\begin{align*}
N(\lambda) = \sum_{\{n : e^{n^2} \leq \lambda\}} 1 \approx \sqrt{\log \lambda}
\end{align*}
and the function $\sqrt{\log}$ is slowly varying,
Theorem \ref{thm:Hardy-Littlewood} implies
\begin{align*}
\heat_P(t) = \sum_{n=0}^{\infty} e^{-t e^{n^2}} \approx \sqrt{-\log t}, && \text{as } t \downarrow 0.
\end{align*}
\end{example}

However, the assumption of $N$ being a regularly varying function is a non-trivial one and puts limitations on the usefulness of Theorem \ref{thm:Hardy-Littlewood}.

\begin{example}
Let $P$ be an operator with $\lambda_n(P) = 2^{n}$ and $M_n(P) = 2^n$. Then,
\begin{align*}
N(\lambda) = \sum_{n: \, 2^n \leq \lambda} 2^n = \sum_{n=0}^{\lfloor \log_2 \lambda \rfloor} 2^n = 2^{\lfloor \log_2 \lambda \rfloor + 1} - 1 \approx \lambda \cdot F(\lambda)
\end{align*}
but $F(\lambda) \vc 2^{\lfloor \log_2 \lambda \rfloor + 1} \lambda^{-1}$ is \emph{not} slowly varying. Indeed, the limit $\lim_{\lambda \to \infty} F(x \lambda)/F(\lambda)$ exists (and is equal to 1) only if $x = 2^m$ with $m \in \N$.

On the other hand,
\begin{align*}
\zP(s) = \sum_{n = 0}^{\infty} 2^n 2^{-sn} = \frac{1}{1-2^{s-1}}.
\end{align*}
The results of Section \ref{sec:exp} can be easily adapted to this setting yielding the following \emph{absolutely exact} expansion
\begin{align*}
\heat_P(t) = -\frac{t^{-1}}{\log 2} \, \sum_{k \in \Z} \Gamma\left( 1 - \frac{2 \pi i k}{\log 2} \right) t^{\tfrac{2 \pi i k}{\log 2}} + \sum_{n = 0}^{\infty} \frac{(-1)^n}{n!} \, \frac{t^n}{1-2^{-n-1}},
\end{align*}
which is valid for all $t>0$. Note that the leading term of $\heat_P$ is of the form $t^{-1} G(t)$, where $G$ is oscillating and thus not slowly varying.
\end{example}

We have seen that the spectral growth function \eqref{spectral_growth} of an operator is a primary quantity that allows us to determine whether an asymptotic expansion of the associated heat trace can be obtained via the inverse Mellin transform. Proposition \ref{cor:hk_abscissa} establishes an upper bound on the spectral growth of $P$, whereas Proposition \ref{prop:lacunary} gives a lower one. One might therefore expect that for the intermediate values of growth rates, the inverse Mellin transform technique guarantees the existence of an asymptotic expansion of heat traces. However, this is not the case as we shall show below.

\begin{Prop}\label{prop:nonexample}
Let $P$ be an operator such that its spectral growth function satisfies
\begin{align}\label{SG_a}
N(\lambda) \approx (\log \lambda)^a, && \text{ as } \lambda \to \infty,
\end{align}
for some $a \in \R^+$. If  $a \notin \N$, then $\zP$ has an abscissa of convergence $L = 0$, but is not meromorphic around $s=0$.
\end{Prop}
\begin{proof}
Proposition \ref{cor:zeta_abscissa} implies that $\zP$ is well-defined with an abscissa of convergence $L = 0$. To see that $\zP$ is not meromorphic around $s=0$ we invoke the Hardy-Littlewood Tauberian Theorem \ref{thm:Hardy-Littlewood} once again. 

With $M \in \N$ such that $\lambda_M \leq 1$ and $\lambda_{M+1} > 1$, we write the zeta-function $\zP$ on $\R^+$ as Riemann-Stieltjes integral (compare \cite[Section 13.3]{Shubin})
\begin{align*}
\zP(s) & = \int_0^{\infty} \lambda^{-s} \, dN(\lambda), && \trunc{\zeta}_{P}^{M}(s) = \int_1^{\infty} \lambda^{-s} \, dN(\lambda).
\end{align*}
On the strength of \cite[Theorem 11a]{Widder} we can change variables in the Riemann-Stieltjes integral and rewrite
\begin{align*}
\trunc{\zeta}_{P}^{M} = \int_1^{\infty} e^{-s \log \lambda} \, dN(\lambda) = \int_0^{\infty} e^{-s \mu} \, d\wt{N}(\mu),
\end{align*}
for $s>0$, with $\wt{N}(\mu) = \sum_{n : \log \lambda_n \leq \mu} M_n$, which is of bounded variation on $\R_+$. Assumption \eqref{SG_a} implies
\begin{alignat*}{2}
\wt{N}(\mu)&\approx\mu^a, &\qquad&\text{as }\mu\to\infty,\\
\intertext{thus, by Theorem \ref{thm:Hardy-Littlewood}, we have}
\trunc{\zeta}_{P}^{M}(s)&\approx s^{-a}\Gamma(a+1),&&\text{as }s\downarrow0.
\end{alignat*}
Therefore, if an extension of $\trunc{\zeta}_{P}^{M}$ to $\Re(s) \leq 0$ exists at all, then the point $s=0$ is not a pole unless $a \in \Np$. Since $\trunc{\zeta}_{P}^{M}$ and $\zP$ differ by an entire function, the same conclusion holds for $\zP$.
\end{proof}

Let us stress that Proposition \ref{prop:nonexample} has a ``no-go'' character only. Having $a \in \Np$ does not imply that $\zP$ can be extended to the left complex half-plane in a meromorphic way.

Proposition \ref{prop:nonexample} is in accordance with Proposition \ref{prop:lacunary}, but it also provides examples of pathological operators, which are not lacunary.

\begin{example}
Consider an operator $P$ with $\lambda_n = e^{n^{2/3}}$ and no multiplicities (i.e. $M_n\equiv1$). Although its eigenvalues grow slower than exponentially, the zeta-function $\zP$ cannot be meromorphic around $s=0$, as
\begin{align*}
\zP(s) \approx \Gamma(-1/2) s^{-3/2},&&\text{as }s\downarrow0 \text{ along } \Im(s)=0.
\end{align*}
On the other hand, by Theorem \ref{thm:Hardy-Littlewood},
\begin{align*}
\heat_P(t) = \sum_{n=0}^{\infty} e^{-t e^{n^{2/3}}} \approx (-\log t)^{3/2}, && \text{as } t \downarrow 0.
\end{align*}
\end{example}

The above examples show that Tauberian theorems and the inverse Mellin transform have different domains of applicability and can be considered as complementary tools in the study of the asymptotic behaviour of heat traces.

Let us now sum up the results of the paper. In Section \ref{sec:general} we have presented general theorems on the existence and convergence of heat traces associated with positive unbounded operators with compact inverses. The necessary conditions were formulated in terms of spectral zeta-functions. The non-existence of meromorphic extensions of the latter sets a natural limitation of applicability of the inverse Mellin transform. However, lacunary operators seem to be pathological anyway from the viewpoint of noncommutative geometry. For instance, if a spectral triple would have a lacunary Dirac operator, then it would not have a dimension spectrum \cite{ConnesMoscovici}. On the other hand, the assumption \eqref{assum:zeta_b}, even in its more refined version \eqref{hypo_absolute}, seems to be a mild one. In fact, a similar one was adopted in \cite[p. 206]{ConnesMoscovici} ``on the technical side''.

In Section \ref{sec:examples} devoted to examples we always worked with operators, the spectrum of which is known explicitly. Therefore, the operatorial aspect of the problem was somewhat hidden. In practice, one is rarely granted the comfort of knowing the full spectrum of a given operator. Even if this is the case, one would like study the behaviour of the heat trace when a fixed operator $P$ is perturbed to $P + A$, with some bounded $A$. Clearly, a bounded perturbation of $P$ would not change the leading behaviour of $\heat_P(t)$ at small $t$, but it can, at least in principle, spoil the asymptotic expansion.

Indeed, perturbations may drastically change the analytic properties of associated zeta-functions. For instance, the modulus of the Dirac operator on the standard Podle\'s sphere has (up to a multiplicative constant) the following eigenvalues $\lambda_n(|\DD_q|) = q^{-n} - q^n$  (see \cite{dab_sit}). It can thus be considered a sum of an operator of exponential spectrum $P$ and a trace class perturbation $Q$. It turns out that the poles of $\zeta_{P+Q}$ form a regular lattice on the left complex half-plane \cite{PodlesSA}, whereas the poles of $\zP$ are located only on the imaginary axis (see Section \ref{sec:exp}). Although, the convergence properties of the small $t$ expansion of $\heat_{P+Q}(t)$ are not altered by the perturbation $Q$, but the estimates of contour integrals are much more subtle and tedious to control (see \cite[Proposition 4.3]{PodlesSA}).

We regard the investigation of the impact of perturbations on the asymptotic expansion of heat traces to be an important and natural next step in the study initiated in this paper. We hope that a combination of our results with the techniques developed in \cite{DixZeta6,DixZeta4,DixZeta5} can lead to a better understanding of heat traces outside the realm of classical pseudodifferential operators.

A promising application of the results presented in this paper is the possibility of performing exact computations of the spectral action. Let us note that the non-perturbative calculations carried out in \cite{ConnesFLRW,Marcolli1,Marcolli2,Piotrek1,TehPHD} using Poisson summation formula are, according to the nomenclature adopted in Definition \ref{def:almost}, only \emph{almost exact}. However, with the help of Theorem \ref{cor:hk_exact} one is actually able to get exact formulae via (distributional) Laplace transform (see \cite{PhD,PodlesSA}). This technique may prove potentially useful in the study of cosmic topology \cite{Marcolli1}.

\section*{Acknowledgements}

The authors would like to thank Bruno Iochum and Andrzej Sitarz for numerous valuable discussions.

Project operated within the Foundation for Polish Science
IPP Programme``Geometry and Topology in Physical Models''
co-financed by the EU European Regional Development Fund,
Operational Program Innovative Economy 2007-2013

\bibliographystyle{plain}
\bibliography{heatBIB}{}

\begin{thebibliography}{10}

\bibitem{Abram}
Milton Abramowitz and Irene~A Stegun.
\newblock {\em Handbook of Mathematical Functions: with Formulas, Graphs, and
  Mathematical Tables}.
\newblock Courier Dover Publications, 2012.

\bibitem{HeatQFT}
Ivan~G Avramidi.
\newblock Heat kernel approach in quantum field theory.
\newblock {\em Nuclear Physics B -- Proceedings Supplements}, 104(1–3):3 --
  32, 2002.
\newblock Proceedings of the International Meeting on Quantum Gravity and
  Spectral Geometry.

\bibitem{BarHomogeneous}
Christian B{\"a}r.
\newblock The {D}irac operator on space forms of positive curvature.
\newblock {\em J. Math. Soc. Japan}, 48(1), 1996.

\bibitem{BarSpin}
Christian B{\"a}r.
\newblock Dependence of the {Dirac} spectrum on the spin structure.
\newblock {\em Global analysis and harmonic analysis (Marseille-Luminy, 1999)},
  4:17--33, 2000.

\bibitem{Anomalies}
Reinhold~A Bertlmann.
\newblock {\em Anomalies in Quantum Field Theory}, volume~91 of {\em
  International Series of Monographs on Physics}.
\newblock Clarendon Press, 1996.

\bibitem{SphereDirac}
Roberto Camporesi and Atsushi Higuchi.
\newblock On the eigenfunctions of the {Dirac} operator on spheres and real
  hyperbolic spaces.
\newblock {\em Journal of Geometry and Physics}, 20(1):1--18, 1996.

\bibitem{DixZeta6}
Alan Carey and Fedor Sukochev.
\newblock Measurable operators and the asymptotics of heat kernels and zeta
  functions.
\newblock {\em Journal of Functional Analysis}, 262(10):4582 -- 4599, 2012.

\bibitem{CPRS}
Alan~L Carey, John Phillips, Adam Rennie, and Fyodor~A Sukochev.
\newblock The local index formula in semifinite {von Neumann algebras I:
  spectral flow}.
\newblock {\em Advances in Mathematics}, 202(2):451--516, 2006.

\bibitem{ConnesSA}
Ali~H Chamseddine and Alain Connes.
\newblock The spectral action principle.
\newblock {\em Communications in Mathematical Physics}, 186(3):731--750, 1997.

\bibitem{ConnesFLRW}
Ali~H Chamseddine and Alain Connes.
\newblock Spectral action for {Robertson--Walker} metrics.
\newblock {\em Journal of High Energy Physics}, 1210:101, 2012.

\bibitem{Christensen2}
Erik Christensen, Cristina Ivan, and Michel~L Lapidus.
\newblock Dirac operators and spectral triples for some fractal sets built on
  curves.
\newblock {\em Advances in Mathematics}, 217(1):42--78, 2008.

\bibitem{Christensen1}
Erik Christensen, Cristina Ivan, and Elmar Schrohe.
\newblock Spectral triples and the geometry of fractals.
\newblock {\em Journal of Noncommutative Geometry}, 6(2):249--274, 2012.

\bibitem{Cipriani}
Fabio Cipriani, Daniele Guido, Tommaso Isola, and Jean-Luc Sauvageot.
\newblock Spectral triples for the {Sierpi\'nski} gasket.
\newblock {\em Journal of Functional Analysis}, 266(8):4809--4869, 2014.

\bibitem{Comtet}
Louis Comtet.
\newblock {\em Advanced Combinatorics: the Art of Finite and Infinite
  Expansions}.
\newblock Springer, 1974.

\bibitem{ConnesNCG}
Alain Connes.
\newblock {\em Noncommutative Geometry}.
\newblock Academic Press, 1995.

\bibitem{ConnesSU2}
Alain Connes.
\newblock Cyclic cohomology, quantum group symmetries and the local index
  formula for {$SU_q(2)$}.
\newblock {\em J. Inst. Math. Jussieu}, 3(1):17--68, 2004.

\bibitem{ConnesIsospectral}
Alain Connes and Giovanni Landi.
\newblock Noncommutative manifolds, the instanton algebra and isospectral
  deformations.
\newblock {\em Communications in Mathematical Physics}, 221(1):141--159, 2001.

\bibitem{ConnesMarcolli}
Alain Connes and Matilde Marcolli.
\newblock {\em Noncommutative Geometry, Quantum Fields and Motives}, volume~55.
\newblock American Mathematical Soc., 2008.

\bibitem{ConnesMoscovici}
Alain Connes and Henri Moscovici.
\newblock The local index formula in noncommutative geometry.
\newblock {\em Geometric \& Functional Analysis GAFA}, 5(2):174--243, 1995.

\bibitem{ConnesModular}
Alain Connes and Henri Moscovici.
\newblock Modular curvature for noncommutative two-tori.
\newblock {\em Journal of the American Mathematical Society}, 27(3):639--684,
  2014.

\bibitem{Copson}
Edward~Thomas Copson.
\newblock {\em Asymptotic Expansions}.
\newblock Number~55 in Cambridge Tracts in Mathematics. Cambridge University
  Press, 1965.

\bibitem{LacunaryDirichlet2}
Ovidiu Costin and Min Huang.
\newblock Behavior of lacunary series at the natural boundary.
\newblock {\em Advances in Mathematics}, 222(4):1370--1404, 2009.

\bibitem{AllPodles}
Ludwik D\k{a}browski, Francesco D'Andrea, Giovanni Landi, and Elmar Wagner.
\newblock Dirac operators on all {Podle\'s} quantum spheres.
\newblock {\em Journal of Noncommutative Geometry}, 1(2):213--239, 2007.

\bibitem{EquatorialPodles}
Ludwik D\k{a}browski, Giovanni Landi, Mario Paschke, and Andrzej Sitarz.
\newblock The spectral geometry of the equatorial podle\'s sphere.
\newblock {\em Comptes Rendus Mathematique}, 340(11):819--822, 2005.

\bibitem{DiracSUq2}
Ludwik D\k{a}browski, Giovanni Landi, Andrzej Sitarz, Walter~van
  \mbox{Suijlekom}, and Joseph~C V\'arilly.
\newblock The dirac operator on {$SU_{q}(2)$}.
\newblock {\em Communications in Mathematical Physics}, 259(3):729--759, 2005.

\bibitem{dab_sit}
Ludwik D\k{a}browski and Andrzej Sitarz.
\newblock Dirac operator on the standard {Podle\'s} quantum sphere.
\newblock {\em Banach Center Publications}, 61:49--58, 2003.

\bibitem{PhD}
Micha\l{} Eckstein.
\newblock {\em Spectral Action -- Beyond the Almost Commutative Geometry}.
\newblock PhD thesis, Jagiellonian University, 2014.

\bibitem{PodlesSA}
Micha\l{} Eckstein, Bruno Iochum, and Andrzej Sitarz.
\newblock Heat trace and spectral action on the standard {Podle\'s} sphere.
\newblock {\em Communications in Mathematical Physics}, 332(2):627--668, 2014.

\bibitem{Elizalde89}
Emili Elizalde and August Romeo.
\newblock Rigorous extension of the proof of zeta-function regularization.
\newblock {\em Physical Review D}, 40(2):436, 1989.

\bibitem{Erdelyi}
Arthur Erd{\'e}lyi.
\newblock {\em Asymptotic Expansions}.
\newblock Courier Dover Publications, 1956.

\bibitem{TorusSA}
Driss Essouabri, Bruno Iochum, Cyril Levy, and Andrzej Sitarz.
\newblock Spectral action on noncommutative torus.
\newblock {\em Journal of Noncommutative Geometry}, 2(1):53--123, 2008.

\bibitem{Fabry}
Eug\`ene Fabry.
\newblock Sur les s\'eries de taylor qui ont une infinit\'e de points
  singuliers.
\newblock {\em Acta Mathematica}, 22(1):65--87, 1899.

\bibitem{Feller}
Willliam Feller.
\newblock {\em An Introduction to Probability Theory and its Applications},
  volume~2.
\newblock John Wiley \& Sons, 2008.

\bibitem{Flajolet}
Philippe Flajolet, Xavier Gourdon, and Philippe Dumas.
\newblock Meilin transforms and asymptotics: harmonic sums.
\newblock {\em Theoretical Computer Science}, 144(1):3--58, 1995.

\bibitem{Fulling}
Stephen~A Fulling.
\newblock {\em Aspects of Quantum Field Theory in Curved Spacetime}, volume~17
  of {\em London Mathematical Society Student Texts}.
\newblock Cambridge University Press, 1989.

\bibitem{HeatNCT}
Victor Gayral, Bruno Iochum, and Dmitri Vassilevich.
\newblock Heat kernel and number theory on {NC-torus}.
\newblock {\em Communications in Mathematical Physics}, 273(2):415--443, 2007.

\bibitem{Wulkenhaar}
Victor Gayral and Raimar Wulkenhaar.
\newblock Spectral geometry of the {Moyal} plane with harmonic propagation.
\newblock {\em Journal of Noncommutative Geometry}, 7(4):939--979, 2013.

\bibitem{GilkeyGrubb}
Peter~B Gilkey and Gerd Grubb.
\newblock Logarithmic terms in asymptotic expansions of heat operator traces.
\newblock {\em Communications in Partial Differential Equations},
  23(5-6):777--792, 1998.

\bibitem{Gilkey1}
Peter~B Gilkey and Domingo Toledo.
\newblock {\em Invariance Theory, the Heat Equation, and the Atiyah-Singer
  Index Theorem}.
\newblock Studies in Advanced Mathematics. CRC press Boca Raton, second
  edition, 1995.

\bibitem{Guido1}
Daniele Guido and Tommaso Isola.
\newblock Dimensions and singular traces for spectral triples, with
  applications to fractals.
\newblock {\em Journal of Functional Analysis}, 203(2):362--400, 2003.

\bibitem{Guido2}
Daniele Guido and Tommaso Isola.
\newblock Dimensions and spectral triples for fractals in $\mathbb{R}^n$.
\newblock \textit{In Advances in Operator Algebras and Mathematical Physics},
  pages 89--108. Theta Series in Advanced Mathematics, 2005.

\bibitem{Hardy_div}
Godfrey~H Hardy.
\newblock {\em Divergent Series}.
\newblock American Mathematical Society, second edition, 1991.

\bibitem{Hardy}
Godfrey~H Hardy and Marcel Riesz.
\newblock {\em The General Theory of Dirichlet's Series}.
\newblock Courier Dover Publications, 2013.

\bibitem{LacunaryDirichlet}
Isidore~I Hirschman and James~A Jenkins.
\newblock On lacunary dirichlet series.
\newblock {\em Proceedings of the American Mathematical Society},
  1(4):512--517, 1950.

\bibitem{ILVGlobal}
Bruno Iochum, Cyril Levy, and Dmitri Vassilevich.
\newblock Global and local aspects of spectral actions.
\newblock {\em Journal of Physics A: Mathematical and Theoretical},
  45(37):374020, 2012.

\bibitem{ILVWeak}
Bruno Iochum, Cyril Levy, and Dmitri Vassilevich.
\newblock Spectral action beyond the weak-field approximation.
\newblock {\em Communications in Mathematical Physics}, 316(3):595--613, 2012.

\bibitem{ILVTorsion}
Bruno Iochum, Cyril Levy, and Dmitri Vassilevich.
\newblock Spectral action for torsion with and without boundaries.
\newblock {\em Communications in Mathematical Physics}, 310(2):367--382, 2012.

\bibitem{Jonquiere}
Alfred {Jonqui\`ere}.
\newblock {Note sur la s\'erie $\sum_{n=1}^{\infty} \frac{x^n}{n^s}$.}
\newblock {\em {Bulletin de la Soci\'et\'e Math\'ematique de France}},
  17:142--152, 1889.

\bibitem{SeniorKaad}
Jens Kaad and Roger Senior.
\newblock A twisted spectral triple for quantum {$SU(2)$}.
\newblock {\em Journal of Geometry and Physics}, 62(4):731--739, 2012.

\bibitem{JapanHeat}
Tomoyuki Kakehi and Tetsuya Masuda.
\newblock Logarithmic divergence of heat kernels on some quantum spaces.
\newblock {\em T\^ohoku Mathematical Journal}, 47(4):595--600, 1995.

\bibitem{Kaneko}
Masanobu Kaneko, Nobushige Kurokawa, and Masato Wakayama.
\newblock A variation of {Euler's} approach to values of the {Riemann} zeta
  function.
\newblock {\em Kyushu Journal of Mathematics}, 57(1):175--192, 2003.

\bibitem{Lesch}
Matthias Lesch.
\newblock On the noncommutative residue for pseudodifferential operators with
  log-polyhomogeneous symbols.
\newblock {\em Annals of Global Analysis and Geometry}, 17(2):151--187, 1999.

\bibitem{DixZeta4}
Steven Lord, Fedor Sukochev, and Dmitriy Zanin.
\newblock {\em Singular Traces: Theory and Applications}, volume~46 of {\em De
  Gruyter Studies in Mathematics}.
\newblock Walter de Gruyter, 2012.

\bibitem{Mandelbrojt}
Szolem Mandelbrojt.
\newblock S{\'e}ries lacunaires.
\newblock In {\em Actualit\'es scientifiques et industrielles}, page 305.
  Paris, 1936.

\bibitem{Lacunary}
Szolem Mandelbrojt and Edward~RC Miles.
\newblock Lacunary functions.
\newblock In {\em The Rice Institute Pamphlet}, volume~14, pages 261--284.
  Houston, Texas: Rice Institute, 1927.
\newblock Available at http://hdl.handle.net/1911/8511.

\bibitem{Marcolli1}
Matilde Marcolli, Elena Pierpaoli, and Kevin Teh.
\newblock The spectral action and cosmic topology.
\newblock {\em Communications in Mathematical Physics}, 304(1):125--174, 2011.

\bibitem{Marcolli2}
Matilde Marcolli, Elena Pierpaoli, and Kevin Teh.
\newblock The coupling of topology and inflation in noncommutative cosmology.
\newblock {\em Communications in Mathematical Physics}, 309(2):341--369, 2012.

\bibitem{MatsumotoWeng}
Kohji Matsumoto and Lin Weng.
\newblock Zeta-functions defined by two polynomials.
\newblock In {\em Number Theoretic Methods}, pages 233--262. Springer, 2002.

\bibitem{NeshTuset}
Sergey Neshveyev and Lars Tuset.
\newblock A local index formula for the quantum sphere.
\newblock {\em Communications in Mathematical Physics}, 254(2):323--341, 2005.

\bibitem{Piotrek1}
Piotr Olczykowski and Andrzej Sitarz.
\newblock On spectral action over {Bieberbach} manifolds.
\newblock {\em Acta Physica Polonica B}, 42(6), 2011.

\bibitem{PalSundar}
Arupkumar Pal and Sobers Sundar.
\newblock Regularity and dimension spectrum of the equivariant spectral triple
  for the odd-dimensional quantum spheres.
\newblock {\em Journal of Noncommutative Geometry}, 4(3):389--439, 2010.

\bibitem{Paris}
Richard~B Paris and David Kaminski.
\newblock {\em Asymptotics and Mellin-Barnes Integrals}, volume~85 of {\em
  Encyclopedia of Mathematics and its Applications}.
\newblock Cambridge University Press, 2001.

\bibitem{Schrohe}
Elmar Schrohe.
\newblock Complex powers of elliptic pseudodifferential operators.
\newblock {\em Integral Equations and Operator Theory}, 9(3):337--354, 1986.

\bibitem{Shubin}
Mikhail~A Shubin.
\newblock {\em Pseudodifferential Operators and Spectral Theory}.
\newblock Springer, 2001.

\bibitem{DixZeta5}
Fedor Sukochev and Dmitriy Zanin.
\newblock $\zeta$-function and heat kernel formulae.
\newblock {\em Journal of Functional Analysis}, 260(8):2451 -- 2482, 2011.

\bibitem{TehPHD}
Kevin Teh.
\newblock {\em Dirac Spectra, Summation Formulae, and the Spectral Action}.
\newblock PhD thesis, California Institute of Technology, 2013.

\bibitem{Trautman}
Andrzej Trautman.
\newblock Spin structures on hypersurfaces and the spectrum of the {Dirac}
  operator on spheres.
\newblock In {\em Spinors, twistors, Clifford algebras and quantum
  deformations}, pages 25--29. Springer, 1993.

\bibitem{VassilevichReport}
Dmitri~V Vassilevich.
\newblock Heat kernel expansion: user's manual.
\newblock {\em Physics Reports}, 388(5):279--360, 2003.

\bibitem{Vassilevich3}
Dmitri~V Vassilevich.
\newblock Noncommutative heat kernel.
\newblock {\em Letters in Mathematical Physics}, 67(3):185--194, 2004.

\bibitem{Vassilevich2}
Dmitri~V Vassilevich.
\newblock Heat trace asymptotics on noncommutative spaces.
\newblock {\em SIGMA}, 3(093):0708--4209, 2007.

\bibitem{Widder}
David~V Widder.
\newblock {\em The Laplace Transform}.
\newblock Princeton University Press, 1946.

\end{thebibliography}

\end{document}